\documentclass[a4paper,11pt]{amsart} 
\usepackage{amsmath,amsxtra,amssymb,latexsym, amscd,amsthm}
\usepackage[mathscr]{eucal}
\usepackage{enumerate}
\usepackage{pict2e}
\usepackage[a4paper]{geometry}
\usepackage{hyperref}

\theoremstyle{plain}
\newtheorem{thm}{Theorem}[section]
\newtheorem{prp}[thm]{Proposition}

\newtheorem{lem}[thm]{Lemma}
\newtheorem*{stoll*}{Stollmann's Lemma}
\newtheorem*{av*}{Spectral Averaging}
\theoremstyle{definition}

\newtheorem*{rem*}{Remark}
\makeatletter
\@addtoreset{equation}{section}

\newcommand{\dd}{\mathrm{d}}

\newcommand{\N}{\mathbb{N}}
\newcommand{\Z}{\mathbb{Z}}
\newcommand{\Q}{\mathbb{Q}}
\newcommand{\R}{\mathbb{R}}

\DeclareMathOperator{\dist}{dist}
\DeclareMathOperator{\expect}{\mathbb{E}}

\DeclareMathOperator{\one}{\mathbf{1}}
\DeclareMathOperator{\prob}{\mathbb{P}}

\DeclareMathOperator{\supp}{supp}

\DeclareMathOperator{\tr}{tr}

\title[Abstract Wegner estimates]{Some Abstract Wegner estimates with applications}
\author{Mostafa \textsc{Sabri}}
\address{Institut de Math\'{e}matiques de Jussieu, Universit\'{e} Paris Diderot Paris 7, B\^atiment Sophie Germain, 75013, Paris, France}
\email{sabri@math.jussieu.fr}
\date{July 9, 2013}

\begin{document}

\begin{abstract}
We prove some abstract Wegner bounds for random self-adjoint operators. Applications include elementary proofs of Wegner estimates for discrete and continuous Anderson Hamiltonians with possibly sparse potentials, as well as Wegner bounds for quantum graphs with random edge length or random vertex coupling. We allow the coupling constants describing the randomness to be correlated and to have quite general distributions.
\end{abstract}

\bibliographystyle{amsplain}

\maketitle

\section{Introduction}    \label{sec:intro}

\textbf{Background.} Wegner estimates for random Schr\"odinger operators have been the subject of active research for the last three decades. Given a random self-adjoint operator $A(\omega)$ with a discrete spectrum $\{ E_j(\omega) \}$ and a fixed interval $I$, the aim is to obtain good bounds on the average number of $E_j(\omega)$ in $I$. Such estimates can be used in a proof of Anderson localization via multiscale analysis, or in the study of the continuity of the integrated density of states. These estimates are named after Wegner's work \cite{W}.

The aim of this paper is to derive some abstract Wegner bounds for some random self-adjoint operators on a Hilbert space, and to apply them afterwards for specific models. This approach proves to be rewarding, if only because it considerably shortens the proof of a Wegner bound for the model at hand. This is not the first attempt to provide abstract bounds; see \cite{CHKN} for a previous one.

\textbf{Results.} The abstract Wegner estimates are stated in Section~\ref{sec:abs} and applied in Section~\ref{sec:app4}. We first obtain optimal bounds on the lattice and non-optimal bounds in the continuum. We allow the potential to be sparse, i.e. make no covering assumption. This includes models with surface and Delone potentials. We then give Wegner bounds for quantum graphs with random edge lengths or random vertex couplings. In each application, we allow the coupling constants entering the randomness to be correlated and only assume that their distributions have no atoms. A comparison with previous results is provided for each application. We conclude the paper with an appendix describing the spectra of Anderson models with half-space potentials. This illustrates the non-triviality of some of our bounds.

\textbf{Notations.} We assume the probability space has the form $(\Omega, \mathfrak{F}, \prob)$, where $\Omega = \mathcal{B}^{\mathcal{I}}$ for some Borel set $\mathcal{B} \subseteq \R$ and some countable index set $\mathcal{I}$. Here $\prob$ is a probability measure on $\Omega$ and $\mathfrak{F}=\mathop \otimes_{\alpha \in \mathcal{I}} \mathfrak{B}$, where $\mathfrak{B}$ is the Borel $\sigma$-algebra of $\mathcal{B}$. By definition, $\mathfrak{F}$ is generated by cylinder sets of the form $\{ \omega=(\omega_{\alpha}) : \omega_{\alpha_1} \in A_1, \ldots, \omega_{\alpha_n} \in A_n \}$, with $\alpha_j \in \mathcal{I}$ and $A_j \in \mathfrak{B}$. Any product space $\mathcal{B}^{\mathcal{I}}$ is assumed to be endowed with the $\sigma$-algebra $\mathfrak{F}=\mathop \otimes_{\alpha \in \mathcal{I}} \mathfrak{B}$, which we shall often omit.

Fix $\alpha \in \mathcal{I}$, let $Y_{\alpha} := \mathcal{B}^{\mathcal{I} \setminus \{\alpha\}}$, $\mathcal{Y}_{\alpha} := \mathop \otimes_{\beta \neq \alpha} \mathfrak{B}$ and denote $\hat{\omega}_{\alpha} := (\omega_{\beta})_{\beta \neq \alpha}$. Define $\tau_{\alpha} : \Omega \to \mathcal{B} \times Y_{\alpha}$ by $\tau_{\alpha}: \omega \mapsto (\omega_{\alpha},\hat{\omega}_{\alpha})$. Then applying \cite[Corollary 10.4.15]{Bog07} to $(\mathcal{B} \times Y_{\alpha}, \mathfrak{B} \mathop \otimes \mathcal{Y}_{\alpha}, \prob \circ \tau_{\alpha}^{-1})$, we may find for each $\hat{\omega}_{\alpha} \in Y_{\alpha}$ a probability measure $\mu_{\hat{\omega}_{\alpha}}$ on $(\mathcal{B},\mathfrak{B})$ such that, if $A \in \mathfrak{F}$ and $A_{\hat{\omega}_{\alpha}} := \{ \omega_{\alpha} : (\omega_{\alpha},\hat{\omega}_{\alpha}) \in \tau_{\alpha}(A) \}$, then the map $\hat{\omega}_{\alpha} \mapsto \mu_{\hat{\omega}_{\alpha}}(A_{\hat{\omega}_{\alpha}})$ is $\mathcal{Y_{\alpha}}$-measurable and $\prob(A) = \int_{Y_{\alpha}} \mu_{\hat{\omega}_{\alpha}}(A_{\hat{\omega}_{\alpha}}) \dd \prob_{Y_{\alpha}}(\hat{\omega}_{\alpha})$. Here $\prob_{Y_{\alpha}} := \prob \circ \pi_{Y_{\alpha}}^{-1}$, where $\pi_{Y_{\alpha}} : \omega \mapsto \hat{\omega}_{\alpha}$. The measures $\mu_{\hat{\omega}_{\alpha}}(B)$ are essentially regular versions of $\prob \big\{ \omega_{\alpha} \in B \big| \hat{\omega}_{\alpha} \big\}$. We will usually omit the identification map $\tau_{\alpha}$ and simply regard elements of $\Omega$ as ordered pairs $(\omega_{\alpha},\hat{\omega}_{\alpha})$, so that $A_{\hat{\omega}_{\alpha}}$ is a section of $A$, $A_{\hat{\omega}_{\alpha}} = \{ \omega_{\alpha} : (\omega_{\alpha},\hat{\omega}_{\alpha}) \in A \}$.

Now fix a finite set $\mathcal{I}_F \subseteq \mathcal{I}$ (e.g. $\mathcal{I} = \Z^d$ and $\mathcal{I}_F \subset \Z^d$ a cube). We express our Wegner bounds in terms of the following modulus of continuity
\begin{equation}
s_F(\prob,\varepsilon) = \max_{\alpha \in \mathcal{I}_F} \expect_{Y_{\alpha}} \Big\{ \sup_{E \in \R} \mu_{\hat{\omega}_{\alpha}}(E,E+ \varepsilon) \Big\} \, . \label{eq:sP0}
\end{equation}

We show in Section~\ref{sub:technical} that for any probability measure $\mu$ on $\R$,
\begin{equation}
\sup_{E \in \R} \mu(E,E+ \varepsilon) = \sup_{E \in \Q} \mu(E,E+ \varepsilon) \, . \label{eq:sP1}
\end{equation}
In particular, $\hat{\omega}_{\alpha} \mapsto \sup_{E \in \R} \mu_{\hat{\omega}_{\alpha}}(E,E+ \varepsilon)$ is $\mathcal{Y}_{\alpha}$-measurable, so $s_F(\prob,\varepsilon)$ is well defined. We also verify that in the special case where $\prob = \mathop \otimes_{\alpha \in \mathcal{I}} \mu_{\alpha}$ for some probability measures $\mu_{\alpha}$ on $\R$, we have $s_F(\prob,\varepsilon) = \max_{\alpha \in \mathcal{I}_F} \sup_{E \in \R} \mu_{\alpha}(E,E+\varepsilon)$.

\textbf{Remark.} Our bounds are useful if the probability measure $\prob$ is continuous. If $\prob = \mathop \otimes_{\alpha \in \mathcal{I}} \mu$, it is sufficient for localization to have $\mu$ H\"older (or even log-H\"older) continuous. This, of course, encompasses the case where $\mu$ has a bounded density $\mu=\rho(\lambda) \dd \lambda$.

We will not treat here random Schr\"odinger operators with a sign-indefinite single-site potential. The reader can find some Wegner estimates for such models in \cite{K95}, \cite{HK01}, \cite{V1}, \cite{V2}, \cite{PTV} and \cite{Kru}, assuming the distribution $\mu$ of the $(\omega_{\alpha})$ has a density. See also the recent survey \cite{EKTV}. For sign-indefinite models on the lattice, the density assumption on $\mu$ can be relaxed if the disorder is large; see \cite[Theorem 1.2]{ESS} and \cite[Proposition 5.1]{ETV} for a related result. For sign-indefinite models in the continuum however, there are to the best of our knowledge no Wegner bounds without the hypothesis that $\mu$ has a density.

\section{Abstract Theorems}    \label{sec:abs}
In the following we give three abstract Wegner estimates. Theorem~\ref{thm:1} is optimal, but is only valid for finite dimensional spaces. It can be applied for example to discrete Schr\"odinger operators on finite cubes $\Lambda$, acting on $\ell^2(\Lambda)$. Theorems~\ref{thm:weg-stol} and \ref{thm:lppvcov} on the other hand are valid in an arbitrary separable Hilbert space, but they are not optimal.

\subsection{Finite dimensional Hilbert spaces}
\subsection*{Hypotheses (A)}
\begin{enumerate}[1)]
\item We fix a probability space $(\Omega,\mathfrak{F},\prob)$ with $\Omega = \mathcal{B}^{\mathcal{I}}$ for some Borel set $\mathcal{B} \subseteq \R$, some countable index set $\mathcal{I}$, and fix a finite-dimensional Hilbert space $\mathcal{H}$. 
\item $H(\omega)$ is a self-adjoint operator on $\mathcal{H}$ for each $\omega \in \Omega$.
\item Fix a bounded interval $I$. There exist a constant $\gamma>0$ and a self-adjoint operator $W$ such that $\prob$-almost surely,
\[ 
\chi_I(H(\omega)) W \chi_I(H(\omega)) \ge \gamma \chi_I(H(\omega)) \, .
\]
\item The operator $W$ takes the form
\[ 
W = \sum_{\alpha \in \mathcal{I}_F} U_{\alpha} \, ,
\] 
for some finite set $\mathcal{I}_F \subseteq \mathcal{I}$, where the $U_{\alpha}$ are self-adjoint operators.
\item Fix an orthonormal basis $\{e_j\}_{j \in J}$ for $\mathcal{H}$. We define $\mathcal{I}_j := \{ \alpha \in \mathcal{I}_F : U_{\alpha} e_j \neq 0 \}$, $C_{\text{fin}} := \max_{j \in J} | \mathcal{I}_j |$ and  $J_{\textup{eff}} := \{ j \in J: U_{\alpha} e_j \neq 0 \text{ for some } \alpha \in \mathcal{I}_F \}$.
\end{enumerate}

Note that one may take $\mathcal{I} = \mathcal{I}_F = J$ and $W= \sum_{j \in J} P_j = \mathrm{Id}$, where $P_j f := \langle f, e_j \rangle e_j$, in which case conditions 3 and 4 hold trivially on any interval with $\gamma=1$ and $C_{\text{fin}} = 1$. For random Schr\"odinger operators, the $U_{\alpha}$ can be the single-site potentials. Condition 3 is sometimes called an \emph{uncertainty principle}, and an efficient criterion to check its validity was established in \cite{BLS}. The constant $\gamma$ often depends on $I$.

The following proposition is the key idea for obtaining optimal Wegner bounds without covering assumptions. It decomposes the trace into local contributions of the $U_{\alpha}$. The proof is given in Section~\ref{sec:proof1}.

\begin{prp}       \label{prp:local}
Suppose that $H(\omega)$ satisfies Hypotheses \emph{(A)} in the interval $I$. Then $\prob$-almost surely,
\[ 
\tr[\chi_I(H(\omega))] \le \gamma^{-2} C_{\textup{fin}} \sum_{j \in J_{\textup{eff}}} \sum_{\alpha \in \mathcal{I}_j} \langle U_{\alpha} \chi_I(H(\omega)) U_{\alpha} e_j, e_j \rangle \, .
\]
\end{prp}

For our first Wegner bound, we need one more hypothesis:

\subsection*{Hypothesis (B)}
$H(\omega)$ satisfies Hypotheses (A). Moreover, given $\omega = (\omega_{\alpha})_{\alpha \in \mathcal{I}} \in \Omega$, $H(\omega)$ has the form
\[ 
H(\omega) = H_1 + \sum_{\alpha \in \mathcal{I}_F} \omega_{\alpha} U_{\alpha},
\]
where $H_1$ is self-adjoint, $U_{\alpha} \ge 0$ and $\|U_{\alpha}\| \le C_U$ for all $\alpha$.

Hence, randomness must appear as an additive perturbation and the $U_{\alpha}$ must be positive operators. The proof of the next theorem is given in Section~\ref{sec:proof2}.

\begin{thm}    \label{thm:1}
Suppose that $H(\omega)$ satisfies Hypotheses \emph{(A)} and \emph{(B)} in the interval $I$. Then $\tr [ \chi_I(H(\omega))]$ is measurable and
\[ 
\expect \{ \tr [ \chi_I(H(\omega))] \} \le C_W \cdot |J_{\textup{eff}} \,| \cdot s_F(\prob,|I|), 
\]
where $C_W := 6 \gamma^{-2} C^2_U C^2_{\textup{fin}}$ and $s_F(\prob,\varepsilon)$ is defined in \textup{(\ref{eq:sP0})}.
\end{thm}

The fact that uncertainty principles imply Wegner bounds was first realized in \cite{CHK0} and \cite{CHK}. There however, the authors considered the spectral projectors $\chi_I(H_1)$. It was later noticed in \cite{St10} that the arguments become simpler if one considers $\chi_I(H(\omega))$, and this idea was used again in \cite{Klein} and \cite{EK}.

It is worthwile to note that if $H_1$ has the special form $H_1 = \sum_{\alpha \in \mathcal{I}_F} c_{\alpha} U_{\alpha}$, that is, if $H(\omega) = \sum_{\alpha \in \mathcal{I}_F} (c_{\alpha}+\omega_{\alpha}) U_{\alpha}$ for some bounded non-random constants $c_{\alpha}$, then analogs of Proposition~\ref{prp:local} and Theorem~\ref{thm:1} hold for intervals not containing $0$, without the need for an uncertainty principle. Such models arise when studying discrete acoustic operators on $\ell^2(\Z^d)$. We refer the reader to \cite{Kitagaki} for details.

\subsection{Separable Hilbert spaces} We now work in the general setting.

Given $\mathcal{C} \subseteq \R$, we say that $f:\mathcal{C}^{\mathcal{I}} \to \R$ is \emph{monotone increasing} (resp. \emph{monotone decreasing}) if $v_{\alpha} \le w_{\alpha}$ for all $\alpha \in \mathcal{I}$ implies $f(v) \le f(w)$ (resp. $f(v) \ge f(w)$). 

\subsection*{Hypotheses (C)}
\begin{enumerate}[1)]
\item We fix a probability space $(\Omega,\mathfrak{F},\prob)$ with $\Omega = \mathcal{C}^{\mathcal{I}}$ for some interval $\mathcal{C} \subseteq \R$ and some countable index set $\mathcal{I}$. We assume $\prob$ has no atoms, more precisely $s_F(\prob,\varepsilon) \to 0$ as $\varepsilon \to 0$. We also fix a separable Hilbert space $\mathcal{H}$.
\item $H(\omega)$ is a self-adjoint operator on $\mathcal{H}$ for each $\omega \in \Omega$. It is bounded from below and has an orthonormal basis of eigenvectors, with eigenvalues $\lambda_1(\omega) \le \lambda_2(\omega) \le \ldots$\,. 
\item Fix an open interval $I$. There exists a number $K$ independent of $\omega$ such that
\[
n > K \implies \lambda_n(\omega) \notin I.
\]
\item Either $\mathcal{D} := D(H(\omega))$ or $\mathcal{D} :=D(\mathfrak{h}^{\omega})$ is independent of $\omega$, where $\mathfrak{h}^{\omega}$ is the form associated with $H(\omega)$. In the first case we define $f_u(\omega):= \langle H(\omega) u, u \rangle$, in the second case we define $f_u(\omega):= \mathfrak{h}^{\omega}[u]$, for $u \in \mathcal{D}$.
\item There exists a finite set $\mathcal{I}_F \subseteq \mathcal{I}$ such that $f_u(\omega)$ only depends on $(\omega_{\alpha})_{\alpha \in \mathcal{I}_F}$. We denote by $\one_F$ the element $\one_F:=(x_{\alpha}) \in \R^{\mathcal{I}}$ with $x_{\alpha} = 1$ for $\alpha \in \mathcal{I}_F$ and $x_{\alpha}=0$ otherwise.

We also assume that there exists $\gamma>0$ such that for every $u \in \mathcal{D}$, $f_u(\omega)$ satisfies one of the following properties for every $\omega \in \Omega$ and $t \ge 0$ such that $\omega - t \cdot \one_F \in \Omega:$
\begin{enumerate}[a.]
\item $f_u$ is monotone increasing and $f_u(\omega) - f_u(\omega - t \cdot \one_F) \ge t \gamma \|u\|^2$.
\item $f_u \in C^1(\Omega)$, $\frac{\partial f_u(\omega)}{\partial \omega_{\alpha}} \ge 0$ $\forall \alpha \in \mathcal{I}_F$ and $\sum_{\alpha \in \mathcal{I}_F} \frac{\partial f_u(\omega)}{\partial \omega_{\alpha}} \ge \gamma \|u\|^2$.
\item $f_u$ is monotone decreasing and $f_u(\omega) - f_u(\omega - t \cdot \one_F) \le - t \gamma \|u\|^2$.
\item $f_u \in C^1(\Omega)$, $\frac{\partial f_u(\omega)}{\partial \omega_{\alpha}} \le 0$ $\forall \alpha \in \mathcal{I}_F$ and $\sum_{\alpha \in \mathcal{I}_F} \frac{\partial f_u(\omega)}{\partial \omega_{\alpha}} \le - \gamma \|u\|^2$.
\end{enumerate}
\end{enumerate}

Here $f_u \in C^1(\Omega)$ if $f_u$ is continuous on $\Omega$ and continuously differentiable on $\mathring{\Omega}$.

Note that Wegner bounds formulated in terms of $s_F(\prob,\varepsilon)$ are useful precisely when $\prob$ has no atoms. We need this assumption for technical reasons in Section~\ref{sub:technical}. Condition 3 typically holds for any $I \subset \R$. If $\mathcal{H}$ is finite dimensional, it is satisfied with $K=\dim \mathcal{H}$ (since there is no eigenvalue with $n>K$). For infinite dimensional spaces, it is satisfied if $H(\omega)$ is bounded from below by a non random operator $H_0$ with a compact resolvent. In this case, $K$ usually depends on $I$. The only ``real'' conditions are 4 and 5. Condition 5 assumes monotonicity and ``diagonal covering'' for $H(\omega)$. We remove the latter restriction in Theorem~\ref{thm:lppvcov}.

In the applications, it will be important that $H(\omega)$ is not supposed to have the form $H(\omega) = H_1 + \sum_{\alpha \in \mathcal{I}} \omega_{\alpha} U_{\alpha}$, and that Hypothesis (C.5.b) is still sufficient to conclude. Let us state the theorem, see Section~\ref{sec:proof3} for a proof, which is based on ideas from \cite{Sto2}.

\begin{thm}      \label{thm:weg-stol}
Suppose that $H(\omega)$ satisfies Hypotheses \emph{(C)} in the interval $I$. Then $\tr [ \chi_I(H(\omega))]$ is $\mathfrak{F}_{\prob}$-measurable, where $\mathfrak{F}_{\prob}$ denotes the $\prob$-completion of $\mathfrak{F}$, and
\[ 
\overline{\expect} \{ \tr [ \chi_I(H(\omega))] \} \le 2 K \cdot |\mathcal{I}_F| \cdot s_F\Big(\prob,\frac{|I|}{\gamma} \Big) \, ,
\]
where $\overline{\expect}$ denotes the extension of $\expect$ to $\mathfrak{F}_{\prob}$ and $s_F(\prob,\varepsilon)$ is defined in \textup{(\ref{eq:sP0})}.
\end{thm}

Note that we have $s_F(\prob,\varepsilon)$ in the RHS; the quantity $s_F(\overline{\prob},\varepsilon)$ has not been defined. For the applications, classic arguments from \cite{KM82} show that $\tr [ \chi_I(H(\omega))]$ is actually $\mathfrak{F}$-measurable, so that $\overline{\expect}$ reduces to $\expect$ in the LHS.

For a random Schr\"odinger operator restricted to a cube $\Lambda$, the constant $K$ comes e.g. from a Weyl law and takes the form $C \cdot |\Lambda|$. The term $|\mathcal{I}_F|$ measures the contribution of the random potential in $\Lambda$, and will be approximately $|\Lambda|$ for standard single-particle systems. Hence, the upper bound is not linear in $|\Lambda|$.

There are mainly two applications for Wegner estimates: the first to prove localization via multiscale analysis, the second to study the continuity of the integrated density of states (IDS) of $H(\omega)$. For the first purpose, Theorem~\ref{thm:weg-stol} is satisfactory because the term $s_F(\prob,\frac{|I|}{\gamma})$ will be very small assuming $\prob = \mathop \otimes \mu$ with $\mu$ (log-)H\"older continuous, so it will completely outweight the terms $K$ and $|\mathcal{I}_F|$. For the study of the IDS however, this theorem is not satisfactory.

It seems the ``bad'' term here is $K$. Indeed, for discrete models with sparse potentials supported in a set $G$, one expects $|\Lambda \cap G|$ in the upper bound (see Section~\ref{sub:discret}), and this is precisely the term $|\mathcal{I}_F|$ in this case, not $K$ which arguably will be $|\Lambda|$.

Theorem~\ref{thm:1} also has the advantage of avoiding the diagonal cover in Hypothesis (C.5) by means of the uncertainty principle. As we show in Section~\ref{sec:proof3}, there is a related counterpart of this idea for Theorem~\ref{thm:weg-stol}. Namely, if one can prove that the eigenvalues $\lambda_n(\omega)$ of $H(\omega)$ are monotone increasing and satisfy
\[
\big( \lambda_n(\omega) - \lambda_n(\omega - t \cdot \one_F) \big) \chi_I(\lambda_n(\omega)) \ge t \gamma \cdot \chi_I(\lambda_n(\omega)),
\]
then Theorem~\ref{thm:weg-stol} is still valid if we only assume Hypotheses (C.1) to (C.3). 

We finally give our last abstract theorem, which is probably the most original result of this section.

\begin{thm}     \label{thm:lppvcov}
Suppose that $H(\omega)$ satisfies  Hypotheses \emph{(C.1)} to \emph{(C.4)} in the interval $I=(E_1,E_2)$, where $E_2<0$. Assume moreover that $\Omega = [q_-,q_+]^{\mathcal{I}}$, fix $q>q_+$, and suppose there exists $\zeta \neq 0$ such that for any $u \in \mathcal{D}$,
\begin{equation}
f_u(\omega) = a(u) - \sum_{\alpha \in \mathcal{I}_F} (q-\omega_{\alpha})^{\zeta} b_{\alpha}(u) \label{eq:lppv}
\end{equation}
for some finite set $\mathcal{I}_F \subseteq \mathcal{I}$ and some constants $a(u)\ge0$ and $b_{\alpha}(u) \ge 0$. Then $\tr[ \chi_I(H(\omega))]$ is $\mathfrak{F}_{\prob}$-measurable and
\[
\overline{\expect} \{ \tr[ \chi_I(H(\omega))] \} \le 2K \cdot | \mathcal{I}_F | \cdot s_F\left(\prob, (q-q_-)\Big((1+\frac{|I|}{|E_2|})^{\frac{1}{|\zeta|}} - 1 \Big) \right).
\]
In particular, if $\zeta=\pm 1$, we have
\[
\overline{\expect} \{ \tr[ \chi_I(H(\omega))] \} \le 2K \cdot | \mathcal{I}_F | \cdot s_F\Big(\prob, \frac{q-q_-}{|E_2|} |I| \Big).
\]
\end{thm}

The proof is given in Section~\ref{sec:proof4}. It uses an idea from \cite{LPPV}, who roughly considered the case $a(u) \equiv 0$ and $\prob = \mathop \otimes \mu_{\alpha}$ with $\mu_{\alpha} = \rho_{\alpha}(\lambda) \dd \lambda$. Both hypotheses were important to their proof, and we overcome this difficulty by generalizing ideas from \cite{Sto2}. Note however that we need (C.4), the argument of \cite{LPPV} holds under a weaker assumption on $D(H(\omega))$.

Of course, the main advantage here in comparison with Theorem~\ref{thm:weg-stol} is that we only suppose $b_{\alpha}(u) \ge 0$. Theorem~\ref{thm:weg-stol} would need a condition like $\sum_{\alpha \in \mathcal{I}_F} b_{\alpha}(u) \ge \gamma \|u\|^2$ for all $u$. The price to pay is that the bound only holds for specific intervals.

In the applications we shall only need the case $\zeta=1$. However the greater generality does not require additional effort, and we believe it could be useful for models not considered here. For example, the case $\zeta=-2$ appears in the model of \cite{LPPV}.

\section{Applications}             \label{sec:app4}
\subsection{Discrete multi-particle models}      \label{sub:discret}
Consider the Hilbert space $\ell^2(\Z^{nd})$, where $n \in \N^{\ast}$ represents the number of particles living in $\Z^d$. Let $\mathcal{B} \subseteq \R$ be a Borel set and consider a probability space $(\Omega,\prob)$, where $\Omega = \mathcal{B}^{\Z^d}$. Given $\omega = (\omega_{\alpha}) \in \Omega$, let 
\[ 
H(\omega) := H_0 + V^{\omega}, \qquad H_0 := -\Delta + V_0 \, ,
\]
where $-\Delta$ is the discrete Laplace operator on $\ell^2(\Z^{nd})$, $V_0$ is a real non-random potential (possibly an interaction) and for $\mathbf{x}=(x_1,\ldots,x_n)\in (\Z^d)^n \equiv \Z^{nd}$,
\[ 
V^{\omega}(\mathbf{x}) = \sum_{1 \le i \le n} V^{\omega}(x_i) = \sum_{1 \le i \le n} \sum_{\alpha \in \mathbb{Z}^d} \omega_{\alpha}u_{\alpha}(x_i) \, .
\]
Since $\prob$ is arbitrary, the $(\omega_{\alpha})$ are allowed to be correlated. We assume $V_0$ is bounded and the $u_{\alpha}:\Z^d \to \R$ satisfy $0 \le u_{\alpha} \le C_u$ for some uniform $C_u\ge0$. We also assume the $u_{\alpha}$ are compactly supported, that is, if for $j \in \Z^d$, $\mathbf{j} \in \Z^{nd}$ and $L \in \N$ we define the cubes
\[
\Lambda^{(1)}_L(j) := \{ x \in \Z^d : \| x - j \|_{\infty} \le L \}, \quad \text{and} \quad \Lambda^{(n)}_L(\mathbf{j}) := \{ \mathbf{x} \in \Z^{nd} : \| \mathbf{x} - \mathbf{j} \|_{\infty} \le L \},
\]
then we assume there exists an $R\ge0$ such that $u_{\alpha}(j) = 0$ for all $j \notin \Lambda^{(1)}_R(\alpha)$.

As $\supp u_{\alpha}$ is compact, we may interchange the sums and write
\[
V^{\omega} = \sum_{\alpha \in \Z^d} \omega_{\alpha} U_{\alpha}, \quad \text{with} \quad U_{\alpha}(x_1,\ldots,x_n) = \sum_{1\le i \le n} u_{\alpha}(x_i). 
\]

\textbf{Example.} A simple and interesting case is when $n=1$, a non-empty set $G \subset \Z^d$ is given, and $u_{\alpha} = \delta_{\alpha}$ inside $G$ and $u_{\alpha} \equiv 0$ outside $G$, where $\delta_{\alpha}$ is the characteristic function of $\{\alpha\}$. In this case,
\[
H(\omega) = H_0 + \sum_{\alpha \in G} \omega_{\alpha} \delta_{\alpha} \, .
\]

For instance, we may take $G = \Z^d \setminus \{0\}$, which gives rise to a non-covering situation. More generally, $G$ could be a \emph{Delone set} (i.e. $\exists K \ge 0$ such that $\forall j \in \Z^d$, the cube $\Lambda^{(1)}_K(j)$ contains at least one point of $G$) or a subspace $\Z^{d_1} \times \{0\}$ of $\Z^d$, in which case one speaks of \emph{surface potentials}.

\textbf{Discussion of the results.}
\begin{enumerate}[\textbullet]
\item In the case of covering, i.e. $u_{\alpha} \ge c \cdot \delta_{\alpha}$ for all $\alpha$, we have an optimal Wegner bound in any interval $I$. This extends \cite[Theorem 2.1]{K} and \cite[Theorem 2.3]{KN} (because we neither assume that $u_{\alpha} = \delta_{\alpha}$ nor that $\prob= \mathop \otimes \mu$ with $\mu=\rho(t) \dd t$) and improves \cite[Theorem 1]{CS} (because our bound is linear in $|\Lambda|$). Note that the arguments of \cite{KN} actually allow for $\prob$ as general as ours. The multiscale analysis also requires two-volume Wegner bounds (cf. \cite[Corollary 2.4]{KN}); we prove these in \cite{Sab0}.
\item If we have no covering and $\Omega = [q_-,q_+]^{\Z^d}$ with $q_-< 0$, i.e. the perturbation can be negative, we obtain Wegner bounds below $E_0 := \inf \sigma(H_0)$. This extends \cite[Theorems 8,13]{KV}, first because we make no regularity assumption on $\prob$, second because our bound is optimal and valid for multi-particles. Our result also extends the optimal bound \cite[Theorem 4.1]{Kitagaki10} because we allow for general $u_{\alpha}$ and $n$.

But is there any spectrum below $E_0$? We show in Section~\ref{sec:spectra} that if $n=1$, if $G \subseteq \Z^d$ contains a half-space, if $V_0$ is periodic and if $H(\omega) = H_0 + \sum_{\alpha \in G} \omega_{\alpha} \delta_{\alpha}$, then $H(\omega)$ has a spectral interval below $E_0$ almost surely, provided that $\prob = \mathop \otimes \mu$ and $\supp \mu \supseteq [a,b]$, $a<b\le0$. This illustrates that our bound is indeed non-trivial\footnote{Note that a single nonzero $\omega_{\alpha} u_{\alpha}$ actually suffices to create a spectral point below $E_0$ if $q_+<q_{\ast}$, $q_{\ast}=q_{\ast}(\|H_0\|,u_{\alpha})$. So our Wegner bound is also useful when the perturbation is highly negative and the spectral bottom is not an isolated point. This is likely to be the case if the operator is ergodic, e.g. $n=1$ and $H(\omega) = H_0 + \sum_{\alpha \in G} \omega_{\alpha} \delta_{\alpha}$, with $G=(M\Z)^d$ and $V_0$ $M$-periodic.}. The advantage here compared to first item is, of course, the fact that we allow $G \neq \Z^d$.
\item If we have no covering and $\Omega = [q_-,q_+]^{\Z^d}$ with $0 \le q_-$, i.e. the perturbation is positive, we obtain optimal Wegner bounds below $E_q:= \inf \sigma(H_0+qW)$ for any $q>q_-$, where $W := \sum_{\alpha} U_{\alpha}$. But again, is there any spectrum below $E_q$? 

The recent preprints \cite{EK} and \cite{R} have the advantage of giving a complete Wegner bound for some operators in this situation. Namely, the paper \cite{EK} assumes that $n=1$, $H(\omega) = H_0 + \sum_{\alpha \in G} \omega_{\alpha} \delta_{\alpha}$, where $G \subseteq \Z^d$ is a Delone set, $\prob = \mathop \otimes \mu_{\alpha}$ and $\supp \mu_{\alpha} \subset [0,M]$. Under some condition (cf. \cite[Eq.(1.13)]{EK}), the authors establish Wegner bounds for intervals near $E_0$, and in contrast to our result, they show that these intervals contain some spectrum of $H(\omega)$ almost surely. A different proof for this Delone Wegner bound can be found in \cite{R}, in the special case where $V_0 \equiv 0$.

To conclude, let us mention that we can actually use the results of \cite{EK} to illustrate that our Wegner bound for positive perturbations is indeed interesting. Namely, if we take $t:=q$ sufficiently large, then \cite[Theorem 1.3]{EK} combined with \cite[Proposition 1.5]{EK} assert that the above Delone operator has $E_q>E_0$ and some spectrum in $[E_0,E_q)$ almost surely. Our Wegner bound is thus nontrivial for $I \subseteq [E_0,E_q)$.
\end{enumerate}

\textbf{Boundary conditions.} Let $\Lambda:=\Lambda^{(n)}_L(\mathbf{x}) \subset \Z^{nd}$ be a cube. The \emph{simple boundary conditions} are obtained by restricting the matrix of $H(\omega)$ to $\Lambda$, i.e. if $(e_{\mathbf{j}})_{\mathbf{j} \in \Z^{nd}}$ is the standard basis of $\ell^2(\Z^{nd})$, then $H^{\mathrm{S}}_{\Lambda}(\omega)(\mathbf{i},\mathbf{j}) = \langle H(\omega) e_{\mathbf{i}},e_{\mathbf{j}} \rangle$ if both $\mathbf{i},\mathbf{j} \in \Lambda$ and $H^{\mathrm{S}}_{\Lambda}(\omega)(\mathbf{i},\mathbf{j})=0$ otherwise. Next, there are the \emph{Dirichlet} $H^{\mathrm{D}}_{\Lambda}(\omega)$ and \emph{Neumann} $H^{\mathrm{N}}_{\Lambda}(\omega)$ boundary conditions, which were introduced in \cite{Si84} to provide analogs for the lattice of the Dirichlet-Neumann bracketing; see \cite[Section 5.2]{K2} for details. The only identity we need is \cite[Eq. (5.42)]{K2}, which asserts that if $H=-\Delta+V$, then $H \le H_{\Lambda}^{\mathrm{D}} \oplus H_{\Lambda^c}^{\mathrm{D}}$. In particular, if $E_0 := \inf \sigma(H)$, $E^{\Lambda,\mathrm{D}}_0 := \inf \sigma(H_{\Lambda}^{\mathrm{D}})$ and $g \in \ell^2(\Lambda^c)$ is identically zero, then
\[
E^{\Lambda,\mathrm{D}}_0 = \inf_{f \in \ell^2(\Lambda),\|f\|=1} \langle H_{\Lambda}^{\mathrm{D}}f,f \rangle + \langle H_{\Lambda^c}^{\mathrm{D}}g,g \rangle \ge \inf_{\varphi \in \ell^2(\Z^{nd}),\|\varphi\|=1} \langle H \varphi, \varphi \rangle = E_0 \, ,
\]
i.e. Dirichlet boundary conditions \emph{shift the spectrum up}.

\textbf{The result.} Let $\mathbf{x}=(x_1,\ldots,x_n)\in(\Z^d)^n$ and $\Lambda^{(n)}_L(\mathbf{x}) \subset \Z^{nd}$. Consider the Hilbert space $\mathcal{H}:=\ell^2\big(\Lambda^{(n)}_L(\mathbf{x})\big)$ with standard basis $(e_{\mathbf{j}})_{\mathbf{j} \in \Lambda^{(n)}_L(\mathbf{x})}$ and the operator $H^{\bullet}_{\Lambda^{(n)}_L(\mathbf{x})}(\omega)$, where $\bullet = \mathrm{S}$, $\mathrm{D}$ or $\mathrm{N}$. Let
\[
W := \sum_{\alpha \in \Z^d} U_{\alpha}, \quad \text{and} \quad
W_{\Lambda_L^{(n)}(\mathbf{x})} := \sum_{\alpha \in \mathcal{I}_F} U_{\alpha}, \quad \text{where } \mathcal{I}_F := \bigcup_{i=1}^n \Lambda_{L+R}^{(1)}(x_i) \, .
\]

We first show in Theorem~\ref{thm:multi-cont} that uncertainty principles imply Wegner bounds, then we give in Lemma~\ref{lem:uncertain} concrete cases in which the uncertainty principle holds.

\begin{thm}      \label{thm:multi-cont}
Let $\Lambda := \Lambda_L^{(n)}(\mathbf{x})$ be a cube and suppose $H^{\bullet}_{\Lambda}(\omega)$ satisfies
\begin{equation}
\chi_I(H^{\bullet}_{\Lambda}(\omega)) W_{\Lambda} \chi_I(H^{\bullet}_{\Lambda}(\omega)) \ge \gamma \chi_I(H^{\bullet}_{\Lambda}(\omega)) \qquad \prob\text{-a.s.} \label{eq:ap1}
\end{equation}
in an interval $I$, for some $\gamma >0$. Then
\[ 
\expect \{ \tr[ \chi_I(H^{\bullet}_{\Lambda}(\omega)) ] \} \le C_W \cdot \big| \tilde{\Lambda}^{(n)}_L \big| \cdot s_F(\prob,|I|), 
\]
where $C_W = 6 n^4 \gamma^{-2} C^2_u (2R+1)^{2d}$ and $\tilde{\Lambda}^{(n)}_L := \{ \mathbf{j} \in \Lambda^{(n)}_L(\mathbf{x}) : U_{\alpha} e_{\mathbf{j}} \neq 0 \text{ for some } \alpha \in \mathcal{I}_F \}$.
\end{thm}
If $u_{\alpha}=c_{\alpha} \delta_{\alpha}$ with $c_{\alpha} \ge 0$, then $C_u=\sup_{\alpha \in \mathcal{I}} c_{\alpha}$ and $R=0$. If, moreover $n=1$ and $H(\omega) = H_0 + \sum_{\alpha \in G} \omega_{\alpha} \delta_{\alpha}$, then $\tilde{\Lambda}_L^{(1)} = \Lambda_L^{(1)}(x) \cap G$.
\begin{proof}
$H^{\bullet}_{\Lambda}(\omega)$ is a self-adjoint operator given by $H^{\bullet}_{\Lambda}(\omega) = H_1 + \sum_{\alpha \in \mathcal{I}_F} \omega_{\alpha} U_{\alpha}$, with $H_1 = H^{\bullet}_{0,\Lambda}$ self-adjoint. Moreover, $U_{\alpha} \ge 0$, $\| U_{\alpha} \| \le C_U := nC_u$ and $\mathcal{I}_{\mathbf{j}} := \{ \alpha : U_{\alpha} e_{\mathbf{j}} \neq 0 \} \subseteq \bigcup_{k=1}^n \Lambda^{(1)}_R(j_k)$, hence $C_{\text{fin}} := \max |\mathcal{I}_{\mathbf{j}}| \le n (2R+1)^d$. The claim now follows from Theorem~\ref{thm:1}.
\end{proof}

\begin{lem}    \label{lem:uncertain}
Fix $\eta>0$. The uncertainty principle \textup{(\ref{eq:ap1})} holds in any interval
\begin{enumerate}[\rm (1)]
\item $I \subset \R$, if \,$\exists c>0$ such that $u_{\alpha} \ge c \cdot \delta_{\alpha}$ for all $\alpha$, with $\gamma = nc$.
\item $I \subset (- \infty, E_q - \eta]$, if $\Omega = [q_-,q_+]^{\Z^d}$, $q>q_-$ and $E_q := \inf \sigma(H_0 + q W)$, for the Dirichlet restriction $H_{\Lambda}^{\mathrm{D}}$, with $\gamma \ge \frac{\eta}{q-q_-}$.
\end{enumerate}
\end{lem}
Theorem~\ref{thm:multi-cont} combined with Lemma~\ref{lem:uncertain} thus provide a Wegner bound in either situation. If $q_-<0$, we may take $q = 0$ and obtain a Wegner bound below $E_0 := \inf \sigma(H_0)$. Otherwise, $0 \le q_-<q$, and the bound is interesting if $E_q>E_0$, for $I \subseteq [E_0, E_q)$.
\begin{proof}
For (1), note that if $u_{\alpha} \ge c \cdot \delta_{\alpha}$, then for any $\mathbf{y} \in \Lambda_L^{(n)}(\mathbf{x})$,
\[
W_{\Lambda}(\mathbf{y}) \ge c \sum_{1 \le i \le n} \sum_{\alpha \in \mathcal{I}_F} \delta_{\alpha}(y_i) = c \sum_{1 \le i \le n} 1 = nc,
\]
so that $W_{\Lambda} \ge nc$ and (\ref{eq:ap1}) holds trivially in any interval with $\gamma = nc$.

For (2), let $H_q := H_0 + qW$ and given $\omega \in \Omega$, let $\lambda_{\omega}(t) := \inf \sigma(H^{\mathrm{D}}_{\Lambda}(\omega) + tW_{\Lambda})$. Then for any $t \ge q-q_-$ we have
\[
\lambda_{\omega}(t) = \inf \sigma (H^{\mathrm{D}}_{q,\Lambda} + V^{\omega}_{\Lambda} + (t-q)W_{\Lambda}) \ge \inf \sigma \Big(H_q + \sum_{\alpha\in\Z^d} (\omega_{\alpha}+t-q) U_{\alpha} \Big) \ge E_q
\]
where we used the fact that Dirichlet boundary conditions shift the spectrum up. Thus, if $I \subset (-\infty, E_q - \eta]$, we get $\lambda_{\omega}(q-q_-) - \max I \ge \eta$. By \cite[Theorem 1.1]{BLS}, (\ref{eq:ap1}) thus holds in $I$ with $\gamma \ge \frac{\eta}{q-q_-}$. 
\end{proof}

\subsection{Continuum multi-particle models}      \label{sub:cover}
Consider the Hilbert space $L^2(\R^{nd})$, where $n \in \N^{\ast}$ represents the number of particles living in $\R^d$. Let $G \subset \R^d$ be a discrete non-empty set such that $\# \{ \Lambda \cap G \} < \infty$ for any bounded $\Lambda \subset \R^d$ and consider a probability space $(\Omega, \prob)$, where $\Omega := [q_-,q_+]^G$ and $\prob$ has no atoms. Given $\omega=(\omega_{\alpha}) \in \Omega$, let
\[
H(\omega) = H_0 + V^{\omega}, \qquad H_0 := - \Delta + V_0 \, ,
\]
where $V_0 \ge v_0$ is a bounded real non-random potential. We can consider more general $H_0$; we only need $H_{0,\Lambda}$ to satisfy a Weyl law, and this is true for $H_0 = (-i\nabla-A)^2 + V_0$ with weak conditions on $A$ and $V_0$; see \cite[Lemma 5]{HKNSV}. Given $\mathbf{x}=(x_1,\ldots,x_n)\in (\R^d)^n \equiv \R^{nd}$,
\[ 
V^{\omega}(\mathbf{x}) = \sum_{1 \le i \le n} V^{\omega}(x_i) = \sum_{1 \le i \le n} \sum_{\alpha \in G} \omega_{\alpha}u_{\alpha}(x_i) \, .
\]

Let $\Lambda_L^{(n)}(\mathbf{x}) := \{ \mathbf{y} \in \R^{nd} : \|\mathbf{y}-\mathbf{x}\|_{\infty} < L \}$. We assume the $u_{\alpha}:\R^d \to \R$ satisfy $0 \le u_{\alpha} \le C_u$ for some uniform $C_u>0$ and $\supp u_{\alpha} \subset \Lambda^{(1)}_R(\alpha)$ for some $R>0$ independent of $\alpha$. This model encompasses sparse potentials such as Delone and surface potentials. Now put
\[
U_{\alpha}(y_1,\ldots,y_n) := \sum_{i=1}^n u_{\alpha}(y_i) \qquad \text{and} \qquad W:= \sum_{\alpha \in G} U_{\alpha} \, .
\]

Given $\mathbf{z} = (z_1,\ldots,z_n) \in (\R^d)^n$ and a cube $\Lambda_L^{(n)}(\mathbf{z}) \subset \R^{nd}$, let $H^{\bullet}_{\Lambda_L^{(n)}(\mathbf{z})}(\omega)$ be a restriction of $H(\omega)$ acting on $\mathcal{H} := L^2 \big(\Lambda_L^{(n)}(\mathbf{z}) \big)$, with $\bullet = \mathrm{D}$, $\mathrm{N}$, $\mathrm{per}$. Note that without a growth condition on $G$, $H(\omega)$ may not be self-adjoint (cf. \cite{KV}), but here we are only concerned with its restriction, which is self-adjoint.

\textbf{Discussion the results.} Our bounds are not linear in $|\Lambda|$, but may be used for localization.
\begin{enumerate}[\textbullet]
\item The covering situation, i.e. when $G = \Z^d$ and $u_{\alpha} \ge c \chi_{\alpha}$ for all $\alpha$, where $\chi_{\alpha}$ is the characteristic function of $[\alpha-\frac{1}{2}, \alpha+\frac{1}{2}]^d$, has already been analyzed in \cite{BCSS} and \cite{KZ}. There the authors proved Wegner bounds in any interval $I \subset \R$ and for arbitrary $\prob$. We do the same here, simply to illustrate Theorem~\ref{thm:weg-stol}.

\item For negative perturbations, we have a Wegner bound below $E_0:= \inf \sigma(H_0)$. This extends \cite[Theorems 8,13]{KV} because we do not impose any regularity on $\prob$, but \cite{KV} has some Wegner bounds which depend linearly on $|\Lambda|$. In the case of surface potentials, i.e. when $G = \Z^{d_1} \times \{0\}$, \cite[Theorem 2.1]{Kitagaki10} provides an optimal bound.

As in the lattice, there is the issue of whether $H(\omega)$ has some spectrum below $E_0$. We show in the Appendix (Section~\ref{sec:spectra}) that continuum single-particle operators with half-space potentials are good examples of operators which have no covering condition and to which we have a non-trivial Wegner bound. 
\item For positive perturbations, we obtain a Wegner bound below $E_q := \inf \sigma(H_0+qW)$, for any $q>q_+$. This result is very close in spirit to \cite[Theorem 2.1]{BLS}, because both are interesting when $E_0$ is a \emph{weak fluctuation boundary}, i.e. when $E_0 < E_q$. Besides the fact that we allow for multi-particles\footnote{Let us mention here that there is a work in progress by Hislop and Klopp in which an optimal Wegner estimate is derived for some non-covering multi-particle Hamiltonians.}, note that our proof is quite elementary. On the other hand, \cite{BLS} builds on the results of \cite{CHK}, which are technically involved, but they provide an optimal bound.

In contrast to the lattice, the question of whether there are interesting operators for which $E_0<E_q$ is well established in the continuum when $n=1$. Already in \cite[Theorem 2.2]{KSS}, it is shown that if $E(t) := \inf \sigma(H_0 + tW)$, then $E(t) - E_0$ grows linearly in $t$, even if the $u_{\alpha}$ have small support, provided $G=\Z^d$ and $V_0$ is periodic. It was later shown in \cite[Sections 4,5]{BNSS} that $E_q>E_0$ for more general operators with surface or Delone potentials, assuming $V_0$ is periodic. In the case of Delone potentials, this result was very recently improved in \cite[Lemma 4.2]{Klein}, namely, it is shown that $E(t) - E_0$ grows linearly in $t$, and $V_0$ is no longer assumed to be periodic.
\end{enumerate}

Much stronger results are known if $n=1$, $G$ is a Delone set, each $u_{\alpha}>0$ in an open set and $\prob = \mathop \otimes_{\alpha \in G} \mu_{\alpha}$. Namely, the Wegner bound of \cite{RV}, which was improved in \cite{Klein}, is valid for any small interval, not just intervals near the spectral bottom. The result of \cite{Klein} also extends the one of \cite{CHK} who considered $G=\Z^d$, but relies on it.

\begin{thm}     \label{thm:multicov}
For any $I=(E_1,E_2)$, there exists $C_W>0$ such that for any cube $\Lambda_L^{(n)}(\mathbf{x})$,
\begin{enumerate}[\rm (1)]
\item If $G = \Z^d$ and $\exists c>0$ with $u_{\alpha} \ge c \cdot \chi_{\alpha}$ for all $\alpha$, where $\chi_{\alpha} := \chi_{[\alpha - \frac{1}{2},\alpha+\frac{1}{2}]^d}$, then 
\[
\expect \{ \tr[ \chi_I(H^{\bullet}_{\Lambda_L^{(n)}(\mathbf{x})})] \} \le C_W \cdot |\Lambda_L^{(n)}(\mathbf{x})| \cdot | \mathcal{I}_F| \cdot s_F\Big(\prob,\frac{|I|}{nc} \Big) \, ,
\]
where $\mathcal{I}_F := \Big(\bigcup_{j=1}^n \Lambda_{L+R}^{(1)}(x_j)\Big) \bigcap G$ and $s_F(\prob,\varepsilon)$ is defined in \textup{(\ref{eq:sP0})}.
\item In the general case, for any $q>q_+$, if $E_2<E_q := \inf \sigma(H_0+qW)$, then
\[
\expect \{ \tr[ \chi_I(H^{\mathrm{D}}_{\Lambda_L^{(n)}(\mathbf{x})})] \} \le C_W \cdot |\Lambda_L^{(n)}(\mathbf{x})| \cdot | \mathcal{I}_F| \cdot s_F\Big(\prob,\frac{q-q_-}{E_q - E_2} |I| \Big) \, .
\]
\end{enumerate}
\end{thm}
Here $C_W = C_W(nd,E_2,v_0)$ if $q_- \ge 0$ and $C_W = C_W(nd,E_2,v_0,nq_-C_uR)$ otherwise.

If $q_+<0$, i.e. the perturbation is negative, we may take $q=0$ and obtain a Wegner bound below $E_0 := \inf \sigma(H_0)$. Otherwise, $0 \le q_+ < q$ and $E_q > E_0$, for many models.
\begin{proof}
Let $\Lambda:=\Lambda_L^{(n)}(\mathbf{x})$. For (1), note that $H^{\bullet}_{\Lambda}(\omega)$ is a self-adjoint operator given by $H^{\bullet}_{\Lambda}(\omega) = H_1 + \sum_{\alpha \in \mathcal{I}_F} \omega_{\alpha} U_{\alpha}$, where $H_1:= H^{\bullet}_{0,\Lambda}$. Given $u \in D(H^{\bullet}_{\Lambda})$, if $f_u(\omega) = \langle H^{\bullet}_{\Lambda}(\omega) u,u \rangle$, then $f_u$ is monotone increasing since $U_{\alpha} \ge 0$. Moreover, if $\mathbf{y} \in \Lambda_L^{(n)}(\mathbf{x})$, then $W_{\Lambda}(\mathbf{y}) := \sum_{\alpha \in \mathcal{I}_F} U_{\alpha}(\mathbf{y}) \ge c \sum_{1 \le i \le n} \sum_{\alpha \in \mathcal{I}_F} \chi_{\alpha}(y_i) = nc$. Hence, $f_u(\omega+t \cdot \one_F) - f_u(\omega) = t \langle W_{\Lambda} u,u \rangle \ge nct \|u\|$. Hypotheses (C) are thus satisfied with $\gamma=nc$, a Weyl constant $K=C|\Lambda|$, and the claim follows from Theorem~\ref{thm:weg-stol}.

For (2), let $A(\omega) := H^{\mathrm{D}}_{\Lambda}(\omega) - E_q$ and $I' = (E_1-E_q,E_2-E_q)$. Then $\chi_I(\lambda) = \chi_{I'}(\lambda-E_q)$, hence $\expect \{ \tr [ \chi_I (H^{\mathrm{D}}_{\Lambda}(\omega)) ] \} = \expect \{ \tr [ \chi_{I'} (A(\omega)) ] \}$. 

Now $A(\omega)$ is a self-adjoint operator given by $A(\omega) = H_1 + \sum_{\alpha \in \mathcal{I}_F} (\omega_{\alpha}-q) U_{\alpha}$, where $H_1:= H^{\mathrm{D}}_{0,\Lambda}+qW_{\Lambda}-E_q$. Since Dirichlet boundary conditions shift the spectrum up, we have $H_1 \ge 0$. Thus, $A(\omega)$ satisfies the hypotheses of Theorem~\ref{thm:lppvcov} in $I'$ with $\zeta=1$, a Weyl constant $K = C |\Lambda|$, and the claim follows since $|I'|=|I|$.
\end{proof}

\subsection{Quantum graphs with random edge length}      \label{sub:length}
Consider the metric graph $(\mathcal{E},\mathcal{V})$ with vertex set $\mathcal{V} = \Z^d$ and edge set $\mathcal{E} = \{ (m,m+h_j): m \in \Z^d, j=1, \ldots, d \}$, where $(h_j)_{j=1}^d$ is the standard basis of $\Z^d$. Each edge $e=(v,v')$ has an initial vertex $\iota e = v$ and a terminal vertex $\tau e = v'$. Now fix $0<l_{\min} < l_{\max} < \infty$ and let $(\Omega,\prob)$ be a probability space, where $\Omega :=[l_{\min},l_{\max}]^{\mathcal{E}}$ and $\prob$ has no atoms. Given $l^{\omega}=(l^{\omega}_e) \in \Omega$, we identify each edge $e$ with $[0,l^{\omega}_e]$, such that $\iota e$ and $\tau e$ correspond to $0$ and $l^{\omega}_e$ respectively, and consider the Hilbert space $\mathcal{H} := \mathop \oplus _{e \in \mathcal{E}} L^2 [0,l^{\omega}_e]$. Fix $\alpha \in \R$ and define the operator
\[
H(l^{\omega},\alpha) :(f_e) \mapsto (- f_e'') \, ,
\] 
\[
D(H(l^{\omega},\alpha)) := \left\{f=(f_e) \in \mathop \oplus \limits_{e \in \mathcal{E}} W^{2,2} (0,l^{\omega}_e) \left| \begin{array} {l} f \text{ is continuous at each} \\ v \in \mathcal{V} \text{ and } f'(v) = \alpha f(v). \\ 
 \end{array} \right. \right\} \, .
\] 

By continuity at $v$, we mean that if $\tau e = \iota b = v$, then $f_e(l^{\omega}_e) = f_b(0) =: f(v)$. Here $f'(v) := \sum \limits_{e: \iota e = v} f_e'(0) - \sum \limits_{e: \tau e = v} f_e'(l^{\omega}_e)$.

Given $L \in \N^{\ast}$, let $\Lambda_L := \{ e \in \mathcal{E} : \| \iota e \|_{\infty} \le L \text{ or } \| \tau e \|_{\infty} \le L \}$ be a cube and put $\mathcal{V}_{\Lambda_L} := \{ \iota e : e \in \Lambda_L \} \cup \{ \tau e : e \in \Lambda_L \}$. This yields a graph $(\Lambda_L,\mathcal{V}_{\Lambda_L})$ and a corresponding operator $H_{\Lambda_L}(l^{\omega},\alpha)$. We denote $H_{\Lambda_L}^{\omega}(\alpha) := H_{\Lambda_L}(l^{\omega},\alpha)$.

\begin{thm}    \label{thm:length}
Let $I \subset (0, \infty)$ be an interval such that $\bar{I} \cap D_0 = \emptyset$, where $D_0 := \bigcup_{k \in \Z} \Big[ \frac{\pi^2 k^2}{l_{\max}^2}, \frac{\pi^2 k^2}{l_{\min}^2} \Big]$. Then there exists $c_1=c_1(d)$ and $c_2=c_2(I)>0$ such that for any interval $J \subset I$ and any cube $\Lambda$ we have
\[
\prob \{ \sigma(H_{\Lambda}^{\omega}(\alpha)) \cap J \neq \emptyset \} \le c_1 \cdot | \Lambda|^2 \cdot s_F\big(\prob, c_2 |J|\big),
\]
where $| \Lambda|$ is the number of edges in $\mathcal{I}_F := \Lambda$ and $s_F(\prob,\varepsilon)$ is as in \textup{(\ref{eq:sP0})}.
\end{thm}

Previous estimates appeared in \cite{LPPV} and \cite{KP2}, both assumed that $\prob = \mathop \otimes_{e \in E} \mu_e$, with $\mu_e = h_e(\lambda) \dd \lambda$, but their bounds were linear in $|\Lambda|$. Our proof heavily relies on the analysis of \cite{KP2}. Our point here is twofold: first, if one makes use of the black box Theorem~\ref{thm:weg-stol}, then a large part of the proof of \cite{KP2} can be omitted, second this allows to extend their localization results in case $\alpha>0$ to $\mu_e$ which are (log-)H\"older continuous.
\begin{proof}
It is proved in \cite[Eq. (9)-(14)]{KP2}, by spectral analytic arguments and without any assumption on $\prob$, that if $E_J$ is the midpoint of $J$, then there exists a discrete random self-adjoint operator $M_{\Lambda}(l^{\omega},E_J)$ acting on $\ell^2(\mathcal{V}_{\Lambda})$ and $b>0$ such that
\[
\prob \{ \sigma(H_{\Lambda}^{\omega}(\alpha)) \cap J \neq \emptyset \} \le \prob \big\{ \dist \big( \sigma(M_{\Lambda}(l^{\omega},E_J)), \alpha \big) \le b |J| \big\} \, .
\]

Moreover, given $u \in \ell^2(\mathcal{V}_{\Lambda})$, the map $l^{\omega} \mapsto f_u( {l^\omega}) := \langle M_{\Lambda}(l^{\omega},E_J) u, u \rangle$ is in $C^1(\Omega)$, only depends on $(l_e^{\omega})_{e \in \Lambda}$ and there exists $\beta>0$ such that
\begin{itemize}
\item $\frac{ \partial M_{\Lambda}(l^{\omega},E_J)}{ \partial l_e^{\omega}} \ge \beta \cdot I^e$ for all $e \in \Lambda$, where $I^e f(v) = f(v)$ if $v \in \{\iota e, \tau e\}$ and $I^e f(v) = 0$ otherwise,
\item and $\sum_{e \in \Lambda} \frac{ \partial M_{\Lambda}(l^{\omega},E_J)}{ \partial l_e^{\omega}} \ge \beta \cdot \sum_{e \in \Lambda} I^e  \ge \beta \cdot \mathrm{Id}_{\ell^2(\mathcal{V}_{\Lambda})}$.
\end{itemize}

Thus, $\frac{\partial f_u(l^{\omega})}{\partial l_e^{\omega}} \ge \beta( |u(\iota e)|^2 + |u(\tau e)|^2) \ge 0$ for $e \in \Lambda$ and $\sum_{e \in \Lambda} \frac{\partial f_u(l^{\omega})}{\partial l_e^{\omega}} \ge \beta \cdot \| u \|^2$. Hence $M_{\Lambda}(l^{\omega},E_J)$ satisfies Hypothesis (C.5.b). Since $\ell^2(\mathcal{V}_{\Lambda})$ is finite dimensional, the rest of Hypotheses (C) are clearly satisfied with $\mathcal{I}_F = \Lambda$ and $K=|\mathcal{V}_{\Lambda}| \le c_d |\Lambda|$. We may thus apply Theorem~\ref{thm:weg-stol} and Markov inequality to get
\begin{align*}
\prob \big\{ \dist \big( \sigma(M_{\Lambda}(l^{\omega},E_J)), \alpha \big) \le b |J| \big\} & = \prob \big\{ \tr \chi_{[\alpha - b|J|, \alpha + b|J|]}(M_{\Lambda}(l^{\omega},E_J)) \ge 1 \big\} \\
& \le 2 c_d |\Lambda|^2 s_F\Big(\prob, \frac{2b}{\beta} |J| \Big) \, . \qedhere
\end{align*}
\end{proof}

\subsection{Quantum graphs with random vertex coupling}
We finally show that Theorem~\ref{thm:lppvcov} can tackle random vertex coupling models without any analytic effort. It seems there are no previous Wegner estimates for such models.

For simplicity consider the graph $(\mathcal{E},\mathcal{V})$ given by $\mathcal{V}=\Z^d$ and $\mathcal{E}$ the set of segments $e=(v,v')$ between the vertices, assigned lengths $l_e$ with $l_{\min} \le l_e \le l_{\max}$. More general structures can be treated similarly. Given $e=(v,v')$, we put $\iota e = v$ and $\tau e = v'$.

Fix $\alpha_-,\alpha_+\in \R$, $\alpha_-<\alpha_+$ and $\emptyset \neq G \subseteq \mathcal{V}$. Let $(\Omega, \prob)$ be a probability space, where $\Omega = [\alpha_-,\alpha_+]^G$, $\prob$ has no atoms and let $\mathcal{H} = \mathop \oplus_{e \in \mathcal{E}} L^2[0,l_e]$. Let $V=(V_e)$ be a bounded real potential, $V \ge c_0$, and given $\alpha^{\omega} = (\alpha^{\omega}_v) \in \Omega$, consider the operator
\[
H(\alpha^{\omega}):(f_e) \mapsto (-f_e''+V_e f_e) \, ,
\]
acting on $(f_e) \in \mathop \oplus_{e \in \mathcal{E}} W^{2,2}(0,l_e)$ which are continuous at all vertices, i.e. $f_e(l_e) = f_b(0) =: f(v)$ if $\tau e = \iota b = v$, and which satisfy
\[
f'(v):= \sum_{e: \iota e = v} f_e'(0) - \sum_{e: \tau e =v} f_e'(l_e) = \begin{cases}
\alpha^{\omega}_v f(v) &\text{if }v \in G\\
0&\text{otherwise}.\end{cases}
\]

The authors in \cite{KP1} studied the case $G = \mathcal{V}$ and established localization for high disorder and near spectral edges using the fractional moments method (which does not rely on Wegner bounds). Their idea was to reduce the problem to one on $\ell^2(\mathcal{V})$, for an associated discrete operator. Below we prove a direct Wegner bound instead.

Given $\Lambda \subseteq \mathcal{E}$, let $\mathcal{V}_{\Lambda} := \{ \iota e : e \in \Lambda \} \cup \{ \tau e : e \in \Lambda \}$ and $\partial \Lambda := \mathcal{V}_{\Lambda} \cap \mathcal{V}_{\Lambda^c}$. Consider the form
\[
\mathfrak{h}^{\omega,\mathrm{D}}_{\Lambda}[f] = \sum_{e \in \Lambda} \big(\|f_e'\|^2_{L^2(0,l_e)} + \langle V_e f_e, f_e \rangle_{L^2(0,l_e)} \big) + \sum_{v \in G \cap \mathcal{V}_{\Lambda}} \alpha^{\omega}_v |f(v)|^2
\]
acting on $(f_e) \in \mathop \oplus_{e \in \Lambda} W^{1,2}(0,l_e)$ which are continuous at $v \in \mathcal{V}_{\Lambda} \setminus \partial \Lambda$ and vanish at $v \in \partial \Lambda$. Note that $\partial \Lambda$ is empty if $\Lambda = \mathcal{E}$. It is known (see \cite{Ku} or \cite[Lemma 4.1]{LPPV}) that $\mathfrak{h}^{\omega,\mathrm{D}}_{\Lambda}$ is closed and bounded from below, and thus corresponds to a self-adjoint operator $H_{\Lambda}^{\mathrm{D}}(\alpha^{\omega})$. Moreover, $H(\alpha^{\omega}) = H_{\mathcal{E}}^{\mathrm{D}}(\alpha^{\omega})$, so we denote $\mathfrak{h}^{\omega}:= \mathfrak{h}_{\mathcal{E}}^{\omega,\mathrm{D}}$.

\begin{lem}           \label{lem:potpourri}
For any $\Lambda \subseteq \mathcal{E}$, $H(\alpha^{\omega}) \le H_{\Lambda}^{\mathrm{D}}(\alpha^{\omega}) \mathop \oplus H_{\Lambda^c}^{\mathrm{D}}(\alpha^{\omega})$. If $\Lambda$ is finite, $H_{\Lambda}^{\mathrm{D}}(\alpha^{\omega})$ has a compact resolvent. Its eigenvalues, denoted $E_j^{\Lambda,\mathrm{D}}$ counting multiplicity, satisfy the following Weyl law: for any $S \in \R$, there exists a non-random $C=C(S,c_0,\alpha_-,l_{\min},l_{\max})$ such that $E_j^{\Lambda,\mathrm{D}} > S$ if $j > C \cdot |\Lambda|$, where $|\Lambda|$ is the number of edges in $\Lambda$.
\end{lem}
\begin{proof}
The bracketing result follows \cite[Lemma 4.2]{LPPV}, namely, $D(\mathfrak{h}_{\Lambda}^{\omega,\mathrm{D}}) \mathop \oplus D(\mathfrak{h}_{\Lambda^c}^{\omega,\mathrm{D}}) \subset D(\mathfrak{h}^{\omega})$ since a function in $D(\mathfrak{h}_{\Lambda}^{\omega,\mathrm{D}}) \mathop \oplus D(\mathfrak{h}_{\Lambda^c}^{\omega,\mathrm{D}})$ is automatically continuous at all $v$. Moreover, if $f = f_1 \mathop \oplus f_2 \in D(\mathfrak{h}_{\Lambda}^{\omega,\mathrm{D}}) \mathop \oplus D(\mathfrak{h}_{\Lambda^c}^{\omega,\mathrm{D}})$, then $\mathfrak{h}_{\Lambda}^{\omega,\mathrm{D}}[f_1] + \mathfrak{h}_{\Lambda^c}^{\omega,\mathrm{D}}[f_2] = \mathfrak{h}^{\omega}[f]$ because $f(v)=0$ on boundary vertices. Thus, $H \le H_{\Lambda}^{\mathrm{D}} \mathop \oplus H_{\Lambda^c}^{\mathrm{D}}$. 

Now suppose $\Lambda$ is finite and as in \cite{EHS}, consider the Neumann-decoupled Laplacian $-\Delta_{\Lambda}^{\text{dec},\mathrm{N}}$ defined via the form $\mathfrak{k}[f] = \sum_{e \in \Lambda} \|f_e'\|^2_{L^2[0,l_e]}$ with $D(\mathfrak{k}) = \mathop \oplus_{e \in \Lambda} W^{1,2}(0,l_e)$. Then $D(\mathfrak{h}^{\omega,\mathrm{D}}_{\Lambda}) \subset D(\mathfrak{k})$ and $\mathfrak{h}^{\omega,\mathrm{D}}_{\Lambda}[f] \ge \mathfrak{k}[f] + c_0\|f\|^2 + \alpha_- \sum_{v \in G \cap \mathcal{V}_{\Lambda}} |f(v)|^2 \ge \frac{1}{2}((\mathfrak{k}+C)[f])$ for some $C=C(l_{\min},l_{\max},\alpha_-,c_0)$ by standard trace estimates, see e.g. \cite[Lemma 8]{Ku}. Thus, $H_{\Lambda}^{\mathrm{D}}(\alpha^{\omega}) \ge \frac{1}{2}(-\Delta_{\Lambda}^{\text{dec},\mathrm{N}}+C)$. But since $-\Delta_{\Lambda}^{\text{dec},\mathrm{N}} = \mathop \oplus_{e \in \Lambda} -\Delta_{(0,l_e)}^{\mathrm{N}}$, its eigenvalues $E_j^{\text{dec},\Lambda}$ are just the eigenvalues $E_k(-\Delta_{(0,l_e)}^{\mathrm{N}}) = \frac{\pi^2 k^2}{4l_e^2}$ with multiplicity $|\Lambda|$. In particular, $E_j^{\text{dec},\Lambda} \to \infty$ as $j \to \infty$, hence $E_j^{\Lambda,\mathrm{D}} \to \infty$ as $j \to \infty$ and $H_{\Lambda}^{\mathrm{D}}(\alpha^{\omega})$ has a compact resolvent by \cite[Theorem XIII.64]{RS}. Moreover, we have $E_j^{\Lambda,\mathrm{D}} \ge \frac{1}{2}(E_j^{\text{dec},\Lambda}+C)$. By the explicit form of $E_j^{\text{dec},\Lambda}$, we know that $E_j^{\text{dec},\Lambda} > 2S-C$ if $j > C_2 |\Lambda|$ for some $C_2 = C_2(l_{\max},S,C)$. Thus, $E_j^{\Lambda,\mathrm{D}} > S$ if $j > C_2 |\Lambda|$ and we are done.
\end{proof}

We may now state our Wegner bound. Fix $q>\alpha_+$ and let $H_0$, $H_q$ be the operators corresponding to $\mathfrak{h}_0[f] = \sum_{e \in \mathcal{E}} \big(\|f_e'\|_{L^2[0,l_e]}^2+\langle V_e f_e, f_e \rangle \big)$ and $\mathfrak{h}_q[f] = \mathfrak{h}_0[f] + q\sum_{v \in G} |f(v)|^2$ respectively, with $D(\mathfrak{h}_0)=D(\mathfrak{h}_q)=D(\mathfrak{h}^{\omega})$. Let $\mathcal{I}_F := G \cap \mathcal{V}_{\Lambda}$ and $s_F(\prob,\varepsilon)$ as in (\ref{eq:sP0}).

\begin{thm}
Let $I=(E_1,E_2)$ be an open interval. 

There exists $C_W=C_W(E_2,c_0,\alpha_-,l_{\min},l_{\max})>0$ such that for any finite $\Lambda \subset \mathcal{E}$ and any $q>\alpha_+$, if $E_2<E_q := \inf \sigma(H_q)$, then
\[
\expect \{ \tr[ \chi_I(H^{\mathrm{D}}_{\Lambda}(\alpha^{\omega}))] \} \le C_W \cdot |\Lambda| \cdot |\mathcal{I}_F| \cdot s_F\Big(\prob, \frac{q-\alpha_-}{E_q-E_2} |I|\Big) \, .
\]
\end{thm}
If $\alpha_+<0$, we may take $q=0$ and obtain a Wegner bound below $E_0 := \inf(\sigma(H_0))$. This result is non-trivial at least when $G = \Z^d$ and the disorder is high, because $H(\alpha^{\omega})$ will have some spectrum below $E_0$ in this case almost surely; see \cite[Theorem 12]{KP1} and the remark thereafter. If $\alpha_+ \ge 0$, the non-triviality will be ensured if $E_0<E_q$.
\begin{proof}
Let $A(\omega) := H^{\mathrm{D}}_{\Lambda}(\alpha^{\omega}) - E_q$ and $I' = (E_1-E_q,E_2-E_q)$. Then $\chi_I(\lambda) = \chi_{I'}(\lambda-E_0)$, hence $\expect \{ \tr [ \chi_I (H^{\mathrm{D}}_{\Lambda}(\alpha^{\omega})) ] \} = \expect \{ \tr [ \chi_{I'} (A(\omega)) ] \}$. Moreover, $A(\omega)$ corresponds to the form $\mathfrak{a}^{\omega}[f] = (\mathfrak{h}_{\Lambda}^{\omega,\mathrm{D}}-E_q)[f]$ with $\mathcal{D} := D(\mathfrak{a}^{\omega}) = D(\mathfrak{h}_{\Lambda}^{\omega,\mathrm{D}})$ non-random, and we have $\mathfrak{a}^{\omega}[f] = \mathfrak{h}_1[f] + \sum_{v \in G \cap \mathcal{V}_{\Lambda}} (\alpha^{\omega}_v-q) |f(v)|^2$, where $\mathfrak{h}_1 := \mathfrak{h}_{q,\Lambda}^{\mathrm{D}} - E_q$. By the bracketing in Lemma~\ref{lem:potpourri}, we have $\mathfrak{h}_1 \ge 0$. Thus, $A(\omega)$ satisfies Hypotheses (C.1) to (C.4) in $I'$, with a Weyl constant $K = C |\Lambda|$ from Lemma~\ref{lem:potpourri}, and the claim follows from Theorem~\ref{thm:lppvcov}.
\end{proof}

\section{Proofs of the general theorems}            \label{sec:proofs}
\subsection{Proof of Proposition~\ref{prp:local}}           \label{sec:proof1}
\begin{proof}
Put $\chi_I := \chi_I(H(\omega))$. By hypothesis, for a.e. $\omega$,
\begin{equation}
\tr [ \chi_I ] \le \gamma^{-1} \tr[ \chi_I W \chi_I ] = \gamma^{-1} \tr[ \chi_I W ] \, .     \label{eq:mah}
\end{equation}
Given $j \in J$ we have
\begin{align*}
\langle \chi_I W e_j, e_j \rangle & = \langle \chi_I W e_j, \chi_I e_j \rangle \le \|\chi_I W e_j\| \cdot \| \chi_I e_j \| \\
& \le \frac{c}{2} \|\chi_I W e_j\|^2 + \frac{1}{2c} \|\chi_I e_j\|^2 = \frac{c}{2} \langle W \chi_I W e_j, e_j \rangle + \frac{1}{2c} \langle \chi_I e_j, e_j \rangle
\end{align*}
for any $c>0$. Summing over $j \in J$ and choosing $c = \gamma^{-1}$ we get by (\ref{eq:mah})
\[
\tr[\chi_I] \le \gamma^{-1} \Big(\frac{\gamma^{-1}}{2} \tr[W \chi_I W] + \frac{1}{2\gamma^{-1}} \tr[ \chi_I] \Big) \, . 
\]
Thus,
\begin{align*}
\tr [ \chi_I ] & \le \gamma^{-2} \tr[W \chi_I W ] \\
& = \gamma^{-2} \sum_{j \in J} \sum_{\alpha, \alpha' \in \mathcal{I}_F} \langle U_{\alpha} \chi_I U_{\alpha'} e_j, e_j \rangle \\
& = \gamma^{-2} \sum_{j \in J} \sum_{\alpha, \alpha' \in \mathcal{I}_j} \langle \chi_I U_{\alpha'} e_j, \chi_I U_{\alpha} e_j \rangle \\
& \le \frac{\gamma^{-2}}{2} \sum_{j \in J} \sum_{\alpha, \alpha' \in \mathcal{I}_j} \big( \| \chi_I U_{\alpha'} e_j \|^2 + \| \chi_I U_{\alpha} e_j \|^2 \big) \\
& \le \gamma^{-2} C_{\text{fin}} \sum_{j \in J} \sum_{\alpha \in \mathcal{I}_j} \| \chi_I U_{\alpha} e_j \|^2 \, .
\end{align*}

This completes the proof, since $\| \chi_I U_{\alpha} e_j \|^2 = \langle U_{\alpha} \chi_I U_{\alpha} e_j, e_j \rangle$, and the terms with $j \notin J_{\text{eff}}$ are zero.
\end{proof}

\subsection{Proof of Theorem~\ref{thm:1}}           \label{sec:proof2}
We first recall \cite[Theorem 3.2]{St10}:

\begin{av*}     \label{thm:av}
Let $\mu$ be a probability measure on $\R$ and $\mathcal{H}$ a Hilbert space. If $A$ is a self-adjoint operator and $0 \le B$ is a bounded operator on $\mathcal{H}$, then for any interval $I$ and any $\phi \in \mathcal{H}$ we have
\[ 
\int_{\R} \langle B^{1/2}\chi_I(A+tB)B^{1/2} \phi , \phi \rangle \dd\mu(t) \le 6 \|B \| \|\phi \|^2 s(\mu,|I|) \, ,
\]
where $s(\mu,\varepsilon) := \sup_{E \in \R} \mu(E,E+\varepsilon)$.
\end{av*}

Note that we could use instead the spectral averaging of \cite{CHK}; in this case the upper bound should be replaced by $4 \|B\|(1+\|B\|) \|\phi\|^2 s(\mu,|I|)$.

The proof in \cite{St10} actually gives $s(\mu,\varepsilon) = \sup_{E \in \R} \mu[E,E+\varepsilon)$, but since
\begin{equation}
\sup_{E \in \R} \mu[E,E+\varepsilon) = \sup_{E \in \R} \mu(E,E+\varepsilon] = \sup_{E \in \R} \mu(E,E+\varepsilon)      \label{eq:sup}
\end{equation}
(see Section~\ref{sub:technical}), the above bound holds.

\begin{proof}[Proof of Theorem~\ref{thm:1}]
To show that $\chi_I(H(\omega))$ is weakly measurable, it suffices to show that $H(\omega)$ is weakly measurable; see \cite{KM82}. Let $\varphi, \psi \in \mathcal{H}$ and let $g(\omega) = \langle H(\omega) \varphi, \psi \rangle = \langle H_0 \varphi, \psi \rangle + \sum_{\alpha \in \mathcal{I}_F} \omega_{\alpha} \langle U_{\alpha} \varphi, \psi \rangle$. Then $g$ only depends on $(\omega_{\alpha})_{\alpha \in \mathcal{I}_F}$, i.e. $\{ \omega : g(\omega) \ge a \} = A \times \mathcal{B}^{\mathcal{I} \setminus \mathcal{I}_F}$ for some $A \subseteq \mathcal{B}^{\mathcal{I}_F}$, so by definition of $\mathfrak{F}$, it suffices to show that $A \in \mathop \otimes_{\alpha \in \mathcal{I}_F} \mathfrak{B}$. In turn, it suffices to show that the map $g_0: \mathcal{B}^{\mathcal{I}_F} \to \R$ given by $g_0: (\omega_{\alpha})_{\alpha \in \mathcal{I}_F} \mapsto \langle H_0 \varphi, \psi \rangle + \sum_{\alpha \in \mathcal{I}_F} \omega_{\alpha} \langle U_{\alpha} \varphi, \psi \rangle$ is Borel measurable, but this is obvious since it is affine. Hence, $\chi_I(H(\omega))$ is weakly measurable and $\tr[ \chi_I(H(\omega))]$ is measurable.

We may thus integrate in Proposition~\ref{prp:local} to get
\[ 
\expect \{ \tr[ \chi_I(H(\omega))] \} \le \gamma^{-2} C_{\text{fin}} \sum_{j \in J_{\text{eff}}} \sum_{\alpha \in \mathcal{I}_j} \expect \{ \langle U_{\alpha} \chi_I(H(\omega)) U_{\alpha} e_j, e_j \rangle \} \, .
\]
Fix $j \in J_{\text{eff}}$, $\alpha \in \mathcal{I}_j$ and put $\phi:= U_{\alpha}^{1/2} e_j$. Then by \cite[Theorem 10.2.1]{Dud},
\[ 
\expect \{ \langle U_{\alpha} \chi_I(H(\omega)) U_{\alpha} e_j, e_j \rangle \} = \expect_{Y_{\alpha}} \Big\{ \int_{\mathcal{B}} \langle U_{\alpha}^{1/2} \chi_I(H(\omega)) U_{\alpha}^{1/2} \phi, \phi \rangle \dd\mu_{\hat{\omega}_{\alpha}}(\omega_{\alpha}) \Big\} \, .
\]
Using the spectral averaging with $A=H_1 + \sum_{\beta \neq \alpha} \omega_{\beta} U_{\beta}$, $B = U_{\alpha}$ and $t= \omega_{\alpha}$, we get
\[ 
\expect \{ \langle U_{\alpha} \chi_I(H(\omega)) U_{\alpha} e_j, e_j \rangle \} \le 6 \| U_{\alpha} \| \|U_{\alpha}^{1/2} e_j \|^2 \expect_{Y_{\alpha}} \{s(\mu_{\hat{\omega}_{\alpha}},|I|) \} \le 6 C^2_U \expect_{Y_{\alpha}}\{ s(\mu_{\hat{\omega}_{\alpha}},|I|) \}. 
\]
Since $\expect_{Y_{\alpha}}\{s(\mu_{\hat{\omega}_{\alpha}},|I|) \} \le s_F(\prob,|I|)$, the proof is complete.
\end{proof} 

\subsection{Proof of Theorem~\ref{thm:weg-stol}}           \label{sec:proof3}
Through this subsection $(\Omega,\mathfrak{F},\prob)$ is a probability space with $\Omega := \mathcal{C}^{\mathcal{I}}$, where $\mathcal{C} \subseteq \R$ is an interval and $\mathcal{I}$ is a countable index set. $\mathfrak{F}_{\prob}$ denotes the $\prob$-completion of $\mathfrak{F}$. We fix a finite set $\mathcal{I}_F \subseteq \mathcal{I}$ and denote by $\one_F$ the element $\one_F=(x_{\alpha}) \in \R^{\mathcal{I}}$ such that $x_{\alpha} = 1$ if $\alpha \in \mathcal{I}_F$ and $x_{\alpha}=0$ otherwise. 

We will use the fact that monotone functions $\varphi:\Omega\to \R$ which depend on finitely many $\omega_{\alpha}$ are $\mathfrak{F}_{\prob}$-measurable; this is proved in Lemma~\ref{lem:monotone}. Note that for any fixed $x \in \R^{\mathcal{I}}$, the map $\varphi(\omega -x)$ is also monotone increasing, hence $\mathfrak{F}_{\prob}$-measurable. We may thus state the following lemma, whose basic idea stem from \cite{Sto2}, see also \cite{C} and \cite{BCSS}.

\begin{lem}      \label{lem:auxi}
Suppose $\varphi : \Omega \to \R$ is monotone increasing and depends on finitely many $\omega_{\alpha}$. Given $c \in \R$ and $\eta>0$, define $A := \{ \omega : \varphi(\omega) \le c \}$, $A^{\eta} := \{ \omega : \omega - \eta \cdot \one_F \in \Omega \text{ and } \varphi(\omega - \eta \cdot \one_F) \le c\}$, $B := \{ \omega : \varphi(\omega) \ge c \}$ and $B^{\eta} := \{ \omega : \omega + \eta \cdot \one_F \in \Omega \text{ and } \varphi(\omega + \eta \cdot \one_F) \ge c\}$. Then
\[
\overline{\prob}(A^{\eta} \setminus A) \le |\mathcal{I}_F| \cdot s_F(\prob, \eta) \quad \text{and} \quad \overline{\prob} (B^{\eta} \setminus B) \le |\mathcal{I}_F| \cdot s_F(\prob, \eta) \, ,
\]
where $\overline{\prob}$ denotes the extension of $\prob$ to $\mathfrak{F}_{\prob}$ and $s_F(\prob,\eta)$ is as in \textup{(\ref{eq:sP0})}.
\end{lem}
\begin{proof}
We prove the second bound; the first is similar. Let $\mathcal{I}_F = \{\alpha_1,\ldots,\alpha_m\}$ and $\mathcal{I}_k = \{\alpha_1,\ldots,\alpha_k \}$ for $1 \le k \le m$. Let $\one_j$ be the element $\one_j = (x_{\alpha}) \in \R^{\mathcal{I}}$ with $x_{\alpha} = 1$ if $\alpha \in \mathcal{I}_k$ and $x_{\alpha} = 0$ otherwise, so that $\one_m = \one_F$. Set
\[
B_0^{\eta} := B \quad \text{and} \quad B_j^{\eta} := \{\omega: \omega + \eta \cdot \one_j \in \Omega \text{ and } \varphi(\omega + \eta \cdot \one_j) \ge c \}
\]
for $1 \le j \le m$. Note that if $B_0,\ldots,B_m$ is any collection of sets, then one checks by induction that $B_m \setminus B_0 \subseteq \bigcup_{j=1}^m (B_j \setminus B_{j-1})$, so we have in particular
\begin{equation}
\overline{\prob} ( B_m^{\eta} \setminus B_0) \le \sum_{j=1}^m \overline{\prob} ( B_j^{\eta} \setminus B_{j-1}^{\eta} ) \, .           \label{eq:wegstol1}
\end{equation}

Now fix $j \in \{1,\ldots, m\}$, let $\hat{\omega}_j = (\omega_{\beta})_{\beta \neq \alpha_j} \in \mathcal{C}^{\mathcal{I} \setminus \{\alpha_j\}}$ and denote by $(x,\hat{\omega}_j)$ the element $(x_{\alpha}) \in \R^{\mathcal{I}}$ with $x_{\alpha_j} = x$ and $x_{\beta} = \omega_{\beta}$ for $\beta \neq \alpha_j$. Define the section
\[
C_{\hat{\omega}_j} := \{ x \in \mathcal{C} : (x,\hat{\omega}_j) \in B_j^{\eta} \setminus B_{j-1}^{\eta} \} = (B_j^{\eta} \setminus B_{j-1}^{\eta})_{\hat{\omega}_{\alpha_j}} \, .
\]

We show that $C_{\hat{\omega}_j}$ is contained in an interval of length $\eta$. If $C_{\hat{\omega}_j} = \emptyset$, this is clear, so suppose $x \in C_{\hat{\omega}_j}$. Fix $\delta \ge \eta$. If $x - \delta \in C_{\hat{\omega}_j}$, then $(x- \delta,\hat{\omega}_j) \in B_j^{\eta}$ and thus $\varphi\big( (x - \delta, \hat{\omega}_j) + \eta \cdot \one_j \big) \ge c$. But $\varphi$ is monotone increasing, so $\varphi\big( (x - \delta, \hat{\omega}_j) + \eta \cdot \one_j \big) = \varphi\big((x - \delta + \eta, \hat{\omega}_j) + \eta \cdot \one_{j-1} \big) \le \varphi\big((x, \hat{\omega}_j) + \eta \cdot \one_{j-1} \big) < c$, since $(x, \hat{\omega}_j) \notin B_{j-1}^{\eta}$\footnote{Note that if $\omega + \eta \cdot \one_j \in \Omega$, then $\omega_{\alpha}+\eta \in \mathcal{C}$ for any $\alpha \in \mathcal{I}_j$, so in particular for any $\alpha \in \mathcal{I}_{j-1}$ and thus $\omega + \eta \cdot \one_{j-1} \in \Omega$.}. This contradiction shows that $x - \delta \notin C_{\hat{\omega}_j}$ for any $\delta \ge \eta$, i.e. $C_{\hat{\omega}_j}$ is contained in a semi-open interval $I_{\hat{\omega}_j}$ of length $\eta$. Let $D_j^{\eta}$ be the set $B_j^{\eta} \setminus B_{j-1}^{\eta}$ with each section $C_{\hat{\omega}_j}$ replaced by $I_{\hat{\omega}_j}$. Then $B_j^{\eta} \setminus B_{j-1}^{\eta} \subseteq D_j^{\eta}$ and $I_{\hat{\omega}_j}$ is a Borel set for any $\hat{\omega}_j$. So applying \cite[Corollary 10.4.15]{Bog07} to $D^{\eta}_j$, taking $Y_j := \mathcal{C}^{\mathcal{I} \setminus \{\alpha_j\}}$ and using (\ref{eq:sup}), we may find $\overline{\mu}_{\hat{\omega}_j}$ such that
\begin{equation}
\overline{\prob}(B_j^{\eta} \setminus B_{j-1}^{\eta}) \le \overline{\expect}_{Y_j}\{ \overline{\mu}_{\hat{\omega}_j}(I_{\hat{\omega}_j}) \} \le \overline{\expect}_{Y_j} \Big\{ \sup_{E \in \R} \overline{\mu}_{\hat{\omega}_j}(E,E+\eta) \Big\} \, .         \label{eq:wegstol2}
\end{equation}
But $\prob\{ \omega_{\alpha_j} \in (E,E+\eta) \} = \expect_{Y_j} \{ \mu_{\hat{\omega}_j}(E,E+\eta) \} = \overline{\expect}_{Y_j} \{ \overline{\mu}_{\hat{\omega}_j}(E,E+\eta) \}$, hence $\mu_{\hat{\omega}_j}(E,E+\eta) = \overline{\mu}_{\hat{\omega}_j}(E,E+\eta)$ outside a $\overline{\prob}_{Y_j}$-null set $\Omega_E$. Let $\Omega_{\ast} = \cup_{E \in \Q} \Omega_E$. Then $\overline{\prob}_{Y_j}(\Omega_{\ast}) = 0$ and $\sup_{E \in \Q} \mu_{\hat{\omega}_j}(E,E+\eta) = \sup_{E \in \Q} \overline{\mu}_{\hat{\omega}_j}(E,E+\eta)$ for any $\hat{\omega}_j \notin \Omega_{\ast}$. So using (\ref{eq:sP1}),
\[
\overline{\expect}_{Y_j} \Big\{ \sup_{E \in \R} \overline{\mu}_{\hat{\omega}_j}(E,E+\eta) \Big\} = \expect_{Y_j} \Big\{ \sup_{E \in \R} \mu_{\hat{\omega}_j}(E,E+\eta) \Big\} \le s_F(\prob,\eta) \, ,
\]
and the claim follows by (\ref{eq:wegstol1}) and (\ref{eq:wegstol2}).
\end{proof}

We may now prove a first extension of Stollmann's Lemma from \cite{Sto2}. Namely, we allow intervals $\mathcal{C}$ and relax the diagonal condition by adding cutoffs $\chi_I(\varphi(\omega))$. The inclusion of cutoffs is actually immediate and will not be used in the proof of Theorem~\ref{thm:weg-stol}. However, this idea plays a major role in the proof of Theorem~\ref{thm:lppvcov}.

\begin{lem}             \label{lem:sto}
Let $I \subset \R$ an open interval. Suppose $\varphi : \Omega \to \R$ is monotone increasing, depends on finitely many $\omega_{\alpha}$ and satisfies
\begin{equation}
\big(\varphi(\omega) - \varphi(\omega - t\cdot \one_F) \big) \chi_I(\varphi(\omega)) \ge t \gamma \cdot \chi_I(\varphi(\omega)) \label{eq:proof0}
\end{equation}
for some $\gamma>0$ and all $t\ge0$ such that $\omega - t\cdot \one_F \in \Omega$. Then
\[ 
\overline{\prob} \{ \varphi(\omega) \in I \} \le \delta \cdot |\mathcal{I}_F| \cdot s_F\Big(\prob, \frac{|I|}{\gamma}\Big), \quad \text{ where } \delta = 
\begin{cases}
1 &\text{if } \inf \mathcal{C}=-\infty,\\
2&\text{otherwise.}\end{cases}
\]
This bound is also true if $\varphi$ is monotone decreasing and satisfies
\begin{equation}
\big(\varphi(\omega) - \varphi(\omega - t\cdot \one_F) \big) \chi_I(\varphi(\omega)) \le - t \gamma \cdot \chi_I(\varphi(\omega)) \, .   \label{eq:proof00}
\end{equation}
\end{lem}
\begin{proof}
Let $I=(a,b)$, $\varepsilon := b-a$ and $\eta := \frac{\varepsilon}{\gamma}$. We have
\[
\overline{\prob} \{\varphi(\omega) \in I \} \le \overline{\prob} \{ \omega : \varphi(\omega) \in I \text{ and } \omega - \eta \cdot \one_F \in \Omega \} + \prob\{ \omega - \eta \cdot \one_F \notin \Omega \} \, .
\]

Put $A := \{ \omega : \varphi(\omega) \le a \}$, $A^{\eta} := \{ \omega : \omega - \eta \cdot \one_F \in \Omega \text{ and } \varphi(\omega - \eta \cdot \one_F) \le a \}$ and let $\omega \in \mathcal{A} :=\{ \omega : \varphi(\omega) \in I \text{ and } \omega - \eta \cdot \one_F \in \Omega \}$. Then by (\ref{eq:proof0}),
\[
\varphi(\omega - \eta \cdot \one_F) \le \varphi(\omega) - \gamma \eta = \varphi(\omega) - \varepsilon \le b - \varepsilon = a \, .
\]
Hence $\omega \in A^{\eta}$. Furthermore, $\varphi(\omega) \in I$ implies $\varphi(\omega) > a$ and thus $\omega \notin A$. Hence $\mathcal{A} \subseteq A^{\eta} \setminus A$ and $\overline{\prob} (\mathcal{A}) \le |\mathcal{I}_F| \cdot s_F(\prob,\eta)$ by Lemma~\ref{lem:auxi}.

If $\inf \mathcal{C} = - \infty$, then $\prob\{ \omega - \eta \cdot \one_F \notin \Omega \} = 0$, since $\mathcal{C}$ is an interval. Otherwise, let $q_-:= \inf \mathcal{C}$. If $q_- \in \mathcal{C}$, then using (\ref{eq:sup}),
\[
\prob\{ \omega - \eta \cdot \one_F \notin \Omega \} = \prob \{ \omega_{\alpha} \in [q_-,q_- + \eta) \text{ for some } \alpha \in \mathcal{I}_F \} \le |\mathcal{I}_F| \cdot s_F(\prob,\eta)
\]
since $\prob\{\omega_{\alpha} \in [q_-,q_-+\eta)\} = \expect_{Y_{\alpha}} \{ \mu_{\hat{\omega}_{\alpha}}[q_-,q_-+\eta)\}$. If $q_- \notin \mathcal{C}$, replace $[q_-,q_- + \eta)$ by $(q_-,q_- + \eta]$.

Finally, if $\varphi$ is decreasing and satisfies (\ref{eq:proof00}), then $\psi:=-\varphi$ is increasing and $\chi_I(\varphi(\omega)) = \chi_{I'}(\psi(\omega))$, where $I':=(-b,-a)$, hence $\psi$ satisfies (\ref{eq:proof0}) in $I'$. Applying the first part we obtain $\overline{\prob} \{ \varphi(\omega) \in I \} = \overline{\prob} \{ \psi(\omega) \in I' \}
\le \delta \cdot |\mathcal{I}_F| \cdot s_F\Big(\prob, \frac{|I'|}{\gamma}\Big)$.
\end{proof} 

\begin{proof}[Proof of Theorem~\ref{thm:weg-stol}]
Let $\{ \varphi_n(\omega) \}$ be an orthonormal basis of eigenvectors of $H(\omega)$ with eigenvalues $\lambda_n(\omega)$. Then $\langle \chi_I(H(\omega)) \varphi_n(\omega) , \varphi_n(\omega) \rangle = \chi_I(\lambda_n(\omega))$. So using (C.3), we get
\[ 
\tr[ \chi_I(H(\omega)) ] = \sum_n \langle \chi_I(H(\omega)) \varphi_n(\omega) , \varphi_n(\omega) \rangle  = \sum_{n\le K} \chi_I(\lambda_n(\omega)) \, .
\]
By (C.4), (C.5) and min-max, each $\lambda_n:\Omega \to \R$ is monotone (see below for (C.5.b) and (C.5.d)) and only depends on $(\omega_{\alpha})_{\alpha \in \mathcal{I}_F}$. So by Lemma~\ref{lem:monotone}, each $\lambda_n$ is $\mathfrak{F}_{\prob}$-measurable, hence $\chi_I(\lambda_n(\omega)) = \chi_{\lambda_n^{-1}(I)}(\omega)$ is $\mathfrak{F}_{\prob}$-measurable, and we may integrate to get
\[ 
\overline{\expect}\{ \tr[\chi_I(H(\omega)) ] \} = \sum_{n\le K} \overline{\expect} \{ \chi_I(\lambda_n(\omega)) \} = \sum_{n \le K} \overline{\prob} \{ \lambda_n(\omega) \in I \} \, .
\]
Now assume (C.5.a) holds. Then by min-max, $\lambda_n(\omega)$ are monotone increasing and satisfy $\lambda_n(\omega) \ge \lambda_n(\omega - t \cdot \one_F) + t \gamma$ for all $t\ge 0$ such that $\omega - t \cdot \one_F \in \Omega$. So by Lemma~\ref{lem:sto},
\[ 
\overline{\prob} \{ \lambda_n(\omega) \in I \} \le 2 \cdot |\mathcal{I}_F| \cdot s_F \Big(\prob, \frac{|I|}{\gamma} \Big), 
\]
as asserted. The case (C.5.c) is similar. Let us show that (C.5.b) implies (C.5.a) and (C.5.d) implies (C.5.c).

Let $f \in C^1(\Omega)$. Given $v,w \in \Omega$, we have $v+t(w-v) \in \Omega$ for any $t \in [0,1]$. Moreover, $t \mapsto f(v+t(w-v))$ is continuous on $[0,1]$ and continuously differentiable on $(0,1)$, hence
\[
f(w)-f(v) = \int_0^1 \frac{\dd}{\dd t} f(v+t(w-v)) \, \dd t = \int_0^1 \sum_{\alpha \in \mathcal{I}} (w_{\alpha}-v_{\alpha}) \frac{\partial f}{\partial \omega_{\alpha}}(v+t(w-v)) \, \dd t. 
\]
If $f$ only depends on $(\omega_{\alpha})_{\alpha \in \mathcal{I}_F}$, then the sum reduces to $\mathcal{I}_F$. If moreover $\frac{\partial f}{\partial \omega_{\alpha}} \ge 0$ on $\Omega$ $\forall \alpha \in \mathcal{I}_F$ and $w_{\alpha} \ge v_{\alpha}$, then $f(w) - f(v) \ge 0$, i.e. $f$ is monotone increasing. Similarly, if $\forall \alpha \in \mathcal{I}_F$, $\frac{\partial f}{\partial \omega_{\alpha}} \le 0$ on $\Omega$, then $f$ is monotone decreasing. Finally, for $w=\omega$ and $v=\omega - \nu \cdot \one_F$ we get
\[ 
f(\omega) - f(\omega - \nu \cdot \one_F) = \nu \int_0^1 \sum_{ \alpha\in \mathcal{I}_F} \frac{\partial f}{\partial \omega_{\alpha}}(\omega - \nu \cdot \one_F +t(\nu \cdot \one_F)) \, \dd t,
\]
hence $\sum_{ \alpha\in \mathcal{I}_F} \frac{\partial f}{\partial \omega_{\alpha}} \ge c$ on $\Omega$ implies $f(\omega) - f(\omega - \nu \cdot \one_F) \ge \nu c$ and $\sum_{ \alpha\in \mathcal{I}_F} \frac{\partial f}{\partial \omega_{\alpha}} \le - c$ on $\Omega$ implies $f(\omega) - f(\omega - \nu \cdot \one_F) \le - \nu c$.
\end{proof}

\subsection{Proof of Theorem~\ref{thm:lppvcov}}           \label{sec:proof4}
The proof of Theorem~\ref{thm:lppvcov} uses two ideas: the first one is roughly to consider the change of variables $v_{\alpha} = \ln \omega_{\alpha}$, so that $\expect \{ f(\omega) \} = \int f(\omega) \prob(\dd \omega) = \int f((e^{v_{\alpha}})) \tilde{\prob}(\dd v)$. This idea was used before in \cite[Theorem 2.9]{LPPV}. The new measure $\tilde{\prob}$ is easily described if $\prob$ is a product measure; the general case is given in Lemma~\ref{lem:push}. The second idea is to generalize Stollmann's lemma to include cutoffs $\chi_I(\varphi(\omega))$ (as we did in Lemma~\ref{lem:sto}) and also extend the diagonal growth condition. This is done in Lemma~\ref{Stol-mod}.

\begin{lem}               \label{lem:push}
Let $\Omega =[q_-,q_+]^{\mathcal{I}}$, fix $q>q_+$ and let $\tilde{\Omega} := [v_-,v_+]^{\mathcal{I}}$, where $v_- = \ln(q-q_+)$ and $v_+ = \ln(q-q_-)$. Define $T : \Omega \to \tilde{\Omega}$ by $T : (\omega_{\alpha}) \mapsto (\ln (q-\omega_{\alpha}))$ and let $\tilde{\prob} := \prob \circ T^{-1}$. Then 
\[
s_F(\tilde{\prob},\varepsilon) \le s_F(\prob,(q-q_-)(e^{\varepsilon} -1)) \, .
\]
\end{lem}
Here $s_F(\tilde{\prob},\varepsilon)$ is defined as before, i.e. if $Z_{\alpha} :=[v_-,v_+]^{\mathcal{I} \setminus \{\alpha \}}$, $\pi_{Z_{\alpha}} : \tilde{\Omega} \to Z_{\alpha}$ is defined by $\pi_{Z_{\alpha}} : v \mapsto \hat{v}_{\alpha}$ and if $\tilde{\prob}_{Z_{\alpha}} = \tilde{\prob} \circ \pi_{Z_{\alpha}}^{-1}$, then $s_F(\tilde{\prob},\varepsilon) = \max_{\alpha \in \mathcal{I}_F} \tilde{\expect}_{Z_{\alpha}} \big\{\sup_{E \in \R} \tilde{\mu}_{\hat{v}_{\alpha}}( E,E+\varepsilon) \big\}$.
\begin{proof}
First recall that by \cite[Theorem 4.1.11]{Dud}, if $\mathcal{T}:(X,\mathcal{X},\prob) \to (Y, \mathcal{Y})$ is any measurable map, and if $\prob^{\mathcal{T}} = \prob \circ \mathcal{T}^{-1}$, then for any measurable $g: Y \to \R$, we have
\begin{equation}
\expect^{\mathcal{T}} \{ g(y) \} = \expect \{ (g \circ \mathcal{T})(x) \} \, ,      \label{eq:change}
\end{equation}
whenever either side exists. Fix $\alpha \in \mathcal{I}_F$ and let $G:= \{ v_{\alpha} \in (E, E+\varepsilon) \}$. Then
\begin{align*}
\tilde{\expect}_{Z_{\alpha}} \{ \tilde{\mu}_{\hat{v}_{\alpha}}(E,E+\varepsilon) \} & = \tilde{\prob}(G) = \prob(T(\omega) \in G) = \prob\{ \ln (q-\omega_{\alpha}) \in (E, E+\varepsilon)\} \\
& = \prob\{\omega_{\alpha} \in (q-e^{E+\varepsilon},q-e^E)\} = \expect_{Y_{\alpha}} \{ \mu_{\hat{\omega}_{\alpha}} (q-e^{E+\varepsilon},q-e^E) \} \, ,
\end{align*}
where $Y_{\alpha} := [q_-,q_+]^{\mathcal{I} \setminus \{ \alpha \}}$. Define $\hat{T}_2 : Z_{\alpha} \to Y_{\alpha}$ by $\hat{T}_2 : (v_{\alpha}) \mapsto (q-e^{v_{\alpha}})$. Then $\hat{T}_2 \circ \pi_{Z_{\alpha}} \circ T = \pi_{Y_{\alpha}}$, so $\tilde{\prob}_{Z_{\alpha}}^{\hat{T}_2} = \prob_{Y_{\alpha}}$ and using (\ref{eq:change}) we get $\expect_{Y_{\alpha}} \{ \mu_{\hat{\omega}_{\alpha}} (q-e^{E+\varepsilon},q-e^E) \} = \tilde{\expect}_{Z_{\alpha}} \{ \mu_{\hat{T_2}(\hat{v}_{\alpha})} (q-e^{E+\varepsilon},q-e^E) \}$. Hence $\tilde{\mu}_{\hat{v}_{\alpha}}( E,E+\varepsilon) = \mu_{\hat{T_2}(\hat{v}_{\alpha})} (q-e^{E+\varepsilon},q-e^E)$ outside a $\tilde{\prob}_{Z_{\alpha}}$-null set $\Omega_E$. Let $\Omega_{\ast} = \cup_{E \in \Q} \Omega_E$. Then $\tilde{\prob}_{Z_{\alpha}}(\Omega_{\ast}) = 0$ and $\sup_{E \in \Q} \tilde{\mu}_{\hat{v}_{\alpha}}( E,E+\varepsilon) = \sup_{E \in \Q} \mu_{\hat{T_2}(\hat{v}_{\alpha})} (q-e^{E+\varepsilon},q-e^E)$ for any $\hat{v}_{\alpha} \notin \Omega_{\ast}$. So using (\ref{eq:sP1}) and (\ref{eq:change}),
\begin{align*}
\tilde{\expect}_{Z_{\alpha}} \Big\{ \sup_{E \in \R} \tilde{\mu}_{\hat{v}_{\alpha}}( E,E+\varepsilon) \Big\} & = \tilde{\expect}_{Z_{\alpha}} \Big\{ \sup_{E \in \R} \mu_{\hat{T_2}(\hat{v}_{\alpha})} (q-e^{E+\varepsilon},q-e^E) \Big\} \\
& = \expect_{Y_{\alpha}} \Big\{ \sup_{E \in \R} \mu_{\hat{\omega}_{\alpha}} (q-e^{E+\varepsilon},q-e^E) \Big\} \, .
\end{align*}
If $q-e^E < q_-$, the RHS is zero, since $\mu_{\hat{\omega}_{\alpha}}$ is supported in $[q_-,q_+]$. So suppose $e^E \le q-q_-$. Then $(q-e^E)-(q-e^{E+ \varepsilon}) = e^E(e^{\varepsilon} -1) \le (q-q_-)(e^{\varepsilon} -1)$. This completes the proof.
\end{proof}

\begin{lem}     \label{Stol-mod}
Let $(\Omega,\prob)$ be a probability space, $\Omega = [c_-,c_+]^{\mathcal{I}}$ and $I \subset \R$ an open interval. Suppose $\varphi : \Omega \to \R$ is monotone increasing, depends on finitely many $\omega_{\alpha}$ and satisfies
\begin{equation}
\big(\varphi(\omega+t\cdot \one_F) - \varphi(\omega) \big) \chi_I(\varphi(\omega)) \ge \gamma (e^{\zeta t}-1) \cdot \chi_I(\varphi(\omega)) \label{eq:proof4}
\end{equation}
for some $\zeta>0$, $\gamma>0$ and all $t\ge0$ such that $\omega+t \cdot \one_F \in \Omega$. Then 
\[
\overline{\prob} \{ \varphi(\omega) \in I \} \le 2 \cdot |\mathcal{I}_F| \cdot s_F\Big(\prob, \frac{1}{\zeta} \ln ( 1+ \frac{|I|}{\gamma}) \Big) \, .
\]
This bound is also true if $\varphi$ is monotone decreasing and satisfies for all $t \ge 0$ such that $\omega - t\cdot \one_F \in \Omega$ the bound
\begin{equation}
\big(\varphi(\omega) - \varphi(\omega - t\cdot \one_F) \big) \chi_I(\varphi(\omega)) \le \gamma (1-e^{\zeta t}) \cdot \chi_I(\varphi(\omega)) \, . \label{eq:proof5}
\end{equation}
\end{lem}
\begin{proof}
Let $I=(a,b)$, $\varepsilon := b-a$ and $\eta := \frac{1}{\zeta} \ln ( 1+ \frac{\varepsilon}{\gamma})$. Suppose first that $\varphi$ is monotone increasing and satisfies (\ref{eq:proof4}). We have
\[
\overline{\prob} \{\varphi(\omega) \in I \} \le \overline{\prob} \{\varphi(\omega) \in I \text{ and } \omega+ \eta \cdot \one_F \in \Omega\} + \prob \{ \omega + \eta \cdot \one_F \notin \Omega \} \, .
\]
For the first term, let $\omega \in \mathcal{A}:= \{\varphi(\omega) \in I \text{ and } \omega+ \eta \cdot \one_F \in \Omega\}$. Then by (\ref{eq:proof4}),
\[
\varphi(\omega + \eta \cdot \one_F) \ge \varphi(\omega) + \gamma (e^{\zeta \eta}-1) = \varphi(\omega) + \varepsilon \ge a + \varepsilon = b \, ,
\]
hence if $B^{\eta} := \{\omega: \omega + \eta \cdot \one_F \in \Omega \text{ and } \varphi(\omega + \eta \cdot \one_F) \ge b\}$, we have $\omega \in B^{\eta}$. Moreover, $\varphi(\omega) \in I$ implies $\varphi(\omega) < b$ and thus $\omega \notin B:=\{ \omega : \varphi(\omega) \ge b \}$. Hence, $\mathcal{A} \subseteq B^{\eta} \setminus B$ and $\overline{\prob}(\mathcal{A}) \le |\mathcal{I}_F| \cdot s_F(\prob, \eta)$ by Lemma~\ref{lem:auxi}. 

For the second term, $\prob \{ \omega + \eta \cdot \one_F \notin \Omega \} = \prob \{ \omega_{\alpha} \in (c_+ - \eta, c_+] \text{ for some } \alpha \in \mathcal{I}_F \} \le |\mathcal{I}_F| \cdot s_F(\prob, \eta)$ by (\ref{eq:sup}). This proves the first claim.

Now suppose $\varphi$ is decreasing and satisfies (\ref{eq:proof5}). Again,
\[
\overline{\prob} \{\varphi(\omega) \in I \} \le \overline{\prob} \{\varphi(\omega) \in I \text{ and } \omega- \eta \cdot \one_F \in \Omega\} + \prob \{ \omega - \eta \cdot \one_F \notin \Omega \} \, .
\]
The second term is assessed as before. For the first term, let $\psi(\omega) := - \varphi(\omega)$ and put $A := \{ \omega : \psi(\omega) \le -b \}$, $A^{\eta} := \{\omega: \omega - \eta \cdot \one_F \in \Omega \text{ and } \varphi(\omega - \eta \cdot \one_F) \le -b\}$ and let $\omega \in \mathcal{A}':=\{\varphi(\omega) \in I \text{ and } \omega- \eta \cdot \one_F \in \Omega\}$. Then by (\ref{eq:proof5}),
\[
\varphi(\omega - \eta \cdot \one_F) \ge \varphi(\omega) - \gamma (1-e^{\zeta \eta}) = \varphi(\omega) + \varepsilon \ge a + \varepsilon = b \, ,
\]
hence $\psi(\omega- \eta \cdot \one_F) \le - b$ and $\omega \in A^{\eta}$. Moreover, $\varphi(\omega) \in I$ implies $\varphi(\omega) < b$, i.e. $\psi(\omega) > -b$ and thus $\omega \notin A$. Hence, $\mathcal{A}' \subseteq A^{\eta} \setminus A$ and the claim follows from Lemma~\ref{lem:auxi}.
\end{proof}

\begin{proof}[Proof of Theorem~\ref{thm:lppvcov}]
Let $A(\omega) = -H(\omega)$ and $I'=(-E_2,-E_1)$. Then $\tr[\chi_I(H(\omega))] = \tr[\chi_{I'}(A(\omega))]$. Moreover, if $r_u(\omega) = - f_u (\omega) = -a(u) + \sum_{\alpha \in \mathcal{I}_F} (q-\omega_{\alpha})^{\zeta} b_{\alpha}(u)$, then using min-max for $H(\omega)$, we obtain the formula
\begin{equation}
\mu_n(\omega) = \inf_{\varphi_1,\ldots,\varphi_{n-1}} \sup_{\substack {u \in \mathcal{D}, \|u \|=1, \\ u \in \{\varphi_1,\ldots,\varphi_{n-1}\}^{\bot}}} r_u(\omega) \label{eq:minmax}
\end{equation}
for the decreasing set $\mu_1(\omega) \ge \mu_2(\omega) \ge \ldots$ of eigenvalues of $A(\omega)$ (here $\mu_j(\omega) = - \lambda_j(\omega)$). Since $b_{\alpha}(u) \ge 0$ for any $u$, each $\mu_n(\omega)$ is monotone and only depends on $(\omega_{\alpha})_{\alpha \in \mathcal{I}_F}$ by (\ref{eq:minmax}), hence each is $\mathfrak{F}_{\prob}$-measurable by Lemma~\ref{lem:monotone}. Thus, as in the proof of Theorem~\ref{thm:weg-stol}, $\tr [ \chi_{I'} (A(\omega)) ]$ is $\mathfrak{F}_{\prob}$-measurable and we may integrate to get
\begin{equation}
\overline{\expect} \{ \tr [ \chi_{I'} (A(\omega)) ] \} = \sum_{n \le K} \overline{\prob} \{ \mu_n(\omega) \in I'\} = \sum_{n \le K} \overline{\tilde{\prob}} \{ \mu_n(T_2(v)) \in I'\} \, , \label{eq:proof1}
\end{equation}
where, using the notations of Lemma~\ref{lem:push}, $T_2 : \tilde{\Omega} \to \Omega$ is given by $T_2 : (v_{\alpha}) \mapsto (q-e^{v_{\alpha}})$, and we applied (\ref{eq:change}) to $g(v) := \chi_{I'}(\mu_n(T_2(v)))$, noting that $(T_2 \circ T)(\omega)=\omega$.

Suppose now that (\ref{eq:lppv}) holds with $\zeta>0$ and fix $u \in \mathcal{D}$. Since $r_u \circ T_2(v) = -a(u) + \sum_{\alpha \in \mathcal{I}_F} e^{\zeta v_{\alpha}}b_{\alpha}(u)$, given $v \in \tilde{\Omega}$ and $t \ge 0$ such that $v+ t \cdot \one_F \in \tilde{\Omega}$, we have
\begin{align*}
(r_u \circ T_2) (v+t \cdot \one_F) & = -a(u) + \sum_{\alpha \in \mathcal{I}_F} e^{\zeta(v_{\alpha}+t)} b_{\alpha}(u) \\
& = - a(u) + e^{\zeta t} \sum_{\alpha \in \mathcal{I}_F} e^{\zeta v_{\alpha}} b_{\alpha}(u) \ge e^{\zeta t} (r_u \circ T_2)(v)
\end{align*}
since $- a(u) \ge - e^{\zeta t} a(u)$. Thus, if $\nu_n(v) := \mu_n(T_2(v))$, we get by (\ref{eq:minmax})
\[
\nu_n(v+t \cdot \one_F) \ge e^{\zeta t} \nu_n(v).
\]
Now note that if $\nu_n(v) \in I'$, then $\nu_n(v) \ge -E_2 =|E_2| >0$. Hence,
\begin{align*}
\big(\nu_n(v+t \cdot \one_F) - \nu_n(v)) \chi_{I'}(\nu_n(v)) & \ge \big(e^{\zeta t} \nu_n(v) - \nu_n(v)) \chi_{I'}(\nu_n(v)) \\
& \ge (e^{\zeta t}-1) |E_2| \chi_{I'}(\nu_n(v)).
\end{align*}
As $\zeta>0$, $\nu_n(v)$ is monotone increasing in $v$, so using Lemma~\ref{Stol-mod} we get
\begin{align*}
\overline{\tilde{\prob}} \{ \nu_n(v) \in I' \} & \le 2\cdot | \mathcal{I}_F | \cdot s_F \Big( \tilde{\prob}, \frac{1}{\zeta} \ln \big(1+\frac{|I'|}{|E_2|}\big) \Big) \\
& \le 2\cdot | \mathcal{I}_F | \cdot s_F \Big( \prob, (q-q_-) \Big(\big(1+\frac{|I|}{|E_2|}\big)^{\frac{1}{\zeta}} -1 \Big)\Big) \, ,
\end{align*}
where we applied Lemma~\ref{lem:push} with $\varepsilon := \frac{1}{\zeta} \ln\big(1+\frac{|I'|}{|E_2|}\big)$. Using (\ref{eq:proof1}), the proof is complete for $\zeta>0$. Now suppose that $\zeta<0$ and put $\theta:=-\zeta>0$. Then
\begin{align*}
(r_u \circ T_2)(v) & = -a(u) + \sum_{\alpha \in \mathcal{I}_F} e^{-\theta v_{\alpha}} b_{\alpha}(u) \\
& = - a(u) + e^{-\theta t} \sum_{\alpha \in \mathcal{I}_F} e^{-\theta (v_{\alpha}-t)} b_{\alpha}(u) \le e^{-\theta t} (r_u \circ T_2)(v-t \cdot \one_F)
\end{align*}
for any $t \ge 0$ such that $v-t \cdot \one_F \in \tilde{\Omega}$, since $-a(u) \le - e^{-\theta t} a(u)$. Hence,
\[
\nu_n(v) \le e^{-\theta t} \nu_n(v-t \cdot \one_F),
\]
and thus, noting that $(1-e^{\theta t}) \le 0$ we get
\begin{align*}
\big(\nu_n(v) - \nu_n(v-t \cdot \one_F)) \chi_{I'}(\nu_n(v)) & \le \big(\nu_n(v) - e^{\theta t} \nu_n(v)) \chi_{I'}(\nu_n(v)) \\
& \le (1-e^{\theta t}) |E_2| \chi_{I'}(\nu_n(v)).
\end{align*}
Furthermore, $\nu_n(v)$ is monotone decreasing. The claim of Theorem~\ref{thm:lppvcov} for $\zeta<0$ now follows as before using Lemma~\ref{Stol-mod}.
\end{proof}

\section{Appendix}       \label{sec:appendix2}

\subsection{Spectra of some Schr\"odinger operators}           \label{sec:spectra}

Let $G \subset \Z^d$ be non-empty, $\mathcal{B} \subseteq \R$ a Borel set and consider the probability space $(\Omega,\prob)$, where $\Omega = \mathcal{B}^G$ and $\prob = \mathop \otimes_{\alpha \in G} \mu$, for some probability measure $\mu$ on $\R$ with $\supp \mu \subseteq \mathcal{B}$. Define
\begin{align*}
H^\omega & = H^0 + V^{\omega} \text{ on } \ell^2(\Z^d), \text{ where } H^0 = - \Delta + V^0, \ V^{\omega} = \sum_{\alpha \in G} \omega_{\alpha} \delta_{\alpha}, \\
H_\omega & = H_0 + V_{\omega} \text{ on } L^2(\R^d), \text{ where } H_0 = - \Delta + V_0, \ V_{\omega} = \sum_{\alpha \in G} \omega_{\alpha} \chi_{\alpha} \, .
\end{align*}
Here $\delta_{\alpha}$ and $\chi_{\alpha}$ are the characteristic functions of $\{\alpha\}$ and $[\alpha-\frac{1}{2},\alpha+\frac{1}{2}]^d$ respectively and $V^0$, $V_0$ are $\Z^d$-periodic bounded real potentials. We denote points in $\R^d$ by $(x^1,\ldots,x^d)$.

We now suppose that $G$ contains a half-space of $\Z^d$, i.e., there exists $r \in \Z$ and $i \in \{1, \ldots, d\}$ such that $(x^1,\ldots,x^d) \in G$ whenever $x^i > r$. Examples are half-spaces of $\Z^d$, and sets with a finite number of holes, i.e. with $\Z^d \setminus G$ finite. We can actually consider more general sets like quarter-spaces or rotated half-spaces. The only thing we need is that $G$ should contain arbitrarily large cubes of $\Z^d$. This excludes $(2\Z)^d$ and thus excludes Delone sets. On the other hand, half-spaces are not Delone sets either since we allow for arbitrarily large cubes with no points of $G$. So the sets we consider here are neither a special case nor a generalization of Delone sets.

\begin{lem}
If $G$ contains a half-space of $\Z^d$, then $\sigma(H^\omega) \supseteq \sigma(H^0) + \supp \mu$ and $\sigma(H_{\omega}) \supseteq \sigma(H_0) + \supp \mu$ almost surely.
\end{lem}
\begin{proof}
We only prove the claim for $H_{\omega}$; the proof is identical for $H^{\omega}$. All the arguments actually go back to \cite{KS80}, \cite{KSS}; one simply needs to choose $\Omega^{\lambda,q}_k(n)$ carefully below.

Assume $(x^1,\ldots,x^d) \in G$ whenever $x^i > r$. Let $E=\lambda+q \in \sigma(H_0) + \supp \mu$. By Weyl's criterion \cite[Theorem 7.22]{Weidmann}, we may find $f_k \in C_c^{\infty}(\R^d)$, $\|f_k\|=1$ such that $\|(H_0-\lambda)f_k\| \to 0$ as $k \to \infty$. Choose $l_k=l_k(\lambda) \in \N^{\ast}$ such that $\supp f_k \subset \Lambda_{l_k}(0)$, put $I_k^q := [q - \frac{1}{k}, q + \frac{1}{k} ]$ and consider the event
\[ 
\Omega^{\lambda,q}_k(n) := \big\{ \omega \in \Omega : \omega_{\alpha} \in I^q_k \quad \forall \alpha \in \Lambda_{l_k}(x_{n,k}) \big\} \, ,
\]
where $x_{n,k} := (3nl_k+r) e^i$ and $e^i \in \Z^d$ has $1$ in the $i$th coordinate and $0$ otherwise. First note that $\Lambda_{l_k}(x_{n,k}) \cap G = \Lambda_{l_k}(x_{n,k})$, so that the above event is well defined. Moreover, $\Lambda_{l_k}(x_{n,k}) \cap \Lambda_{l_k}(x_{m,k}) = \emptyset$ for $n \neq m$, so the events $\{ \Omega^{\lambda,q}_k(n) \}_{n \in \N^{\ast}}$ are independent and $\prob(\Omega^{\lambda,q}_k(n)) = \mu(I^q_k)^{|\Lambda_{l_k}|}$ is the same for all $n$ and strictly positive since $q \in \supp \mu$. It follows by Borel-Cantelli lemma II that if $\Omega^{\lambda,q}_k := \cap_{m \ge 1} \cup_{n \ge m} \Omega^{\lambda,q}_k(n)$, then $\prob(\Omega^{\lambda,q}_k)=1$. Let $\Omega^{\lambda,q} := \cap_{k \in \N^{\ast}} \Omega^{\lambda,q}_k$, then $\prob(\Omega^{\lambda,q}) = 1$.

Now fix $\omega \in \Omega^{\lambda,q}$ and let $k \in \N^{\ast}$. Then $\omega \in \Omega^{\lambda,q}_k$, so we may find $n \in \N^{\ast}$ such that $\omega \in \Omega^{\lambda,q}_k(n)$. But
\[
\|(H_{\omega}-E)f_k(\, \cdot \, - x_{n,k}) \| \le \|(H_0 - \lambda) f_k(\, \cdot \, - x_{n,k}) \| + \|(V_{\omega} - q) f_k(\, \cdot \, - x_{n,k}) \| \, .
\]
Since $V_0$ is periodic, $\|(H_0 - \lambda) f_k(\, \cdot \, - x_{n,k}) \| = \|(H_0 - \lambda) f_k \| \to 0$. Moreover $\omega \in \Omega^{\lambda,q}_k(n)$, so $\omega_{\alpha} \in I^q_k$ for all $\alpha \in \Lambda_{l_k}(x_{n,k})$. Recalling that $\Lambda_{l_k}(x_{n,k}) \cap G = \Lambda_{l_k}(x_{n,k})$, we get
\[
\| (V_{\omega} - q) f_k(\, \cdot \, - x_{n,k}) \|^2 = \sum_{\alpha \in \Lambda_k(x_{n,k})} (\omega_{\alpha}-q)^2\| \chi_{\alpha} f_k \|^2 \le \frac{1}{k^2} \|f_k\|^2 \to 0 \, .
\]
Hence $f_k$ is a Weyl sequence for $E$. We thus showed that for any $\omega \in \Omega^{\lambda,q}$ we have $\lambda + q \in \sigma(H_{\omega})$. Let $\Omega_0 := \bigcap_{\lambda \in \sigma(H_0) \cap \Q, q \in \supp \mu \cap \Q} \Omega^{\lambda,q}$. Then $\prob(\Omega_0) = 1$ and for any $\omega \in \Omega_0$ we have $\sigma(H_{\omega}) \supseteq \sigma(H_0) \cap \Q + \supp \mu \cap \Q$. Since $\sigma(H_{\omega})$ is closed, the proof is complete.
\end{proof}

\subsection{Technical details}             \label{sub:technical}
We give here the details of some claims we made in Sections~\ref{sec:intro} and \ref{sec:proofs}. Let $\mu$ be a probability measure on $\R$. To prove (\ref{eq:sP1}), let $E \in \R$ and $E_k := \frac{ \lfloor 10^k E \rfloor}{10^k}$. Then $E_k \nearrow E$ and $c < E+\varepsilon$ iff $c<E_k + \varepsilon$ for some $k$. Hence
\begin{align*}
\mu(E,E+\varepsilon) = \mu \big( \cup_k (E, E_k + \varepsilon) \big) & = \lim_{k \to \infty} \mu(E, E_k + \varepsilon) \\ 
& \le \lim_{k \to \infty} \mu(E_k, E_k + \varepsilon) \le \sup_{F \in \Q} \mu(F, F+ \varepsilon) \, .
\end{align*}
Thus, $\sup_{E \in \R} \mu(E,E+\varepsilon) \le \sup_{F \in \Q} \mu(F,F+\varepsilon)$. This proves (\ref{eq:sP1}).

Suppose $\prob = \mathop \otimes \mu_{\alpha}$ for some probability measures $\mu_{\alpha}$ on $\R$. Then given $A \in \mathfrak{F}$, we have $\prob(A) = \int_{Y_{\alpha}} \mu_{\alpha}(A_{\hat{\omega}_{\alpha}}) \dd \prob_{Y_{\alpha}}(\hat{\omega}_{\alpha})$, so by \cite[Corollary 10.4.15]{Bog07}, $\mu_{\hat{\omega}_{\alpha}} = \mu_{\alpha}$ $\prob_{Y_{\alpha}}$-a.s., so $s_F(\prob,\varepsilon) = \max_{\alpha \in \mathcal{I}_F} \sup_{E \in \R} \mu_{\alpha}(E,E+\varepsilon)$ using (\ref{eq:sP1}). Next, note that
\begin{align*}
\mu(E,E+\varepsilon] = \mu\big( \cup_k (E+\textstyle{\frac{1}{k}},E+\varepsilon] \big) & = \lim_{k \to \infty} \mu(E+\textstyle{\frac{1}{k}},E+\varepsilon] \\
& \le \lim_{k \to \infty} \mu(E+\textstyle{\frac{1}{k}},E+\varepsilon +\textstyle{\frac{1}{k}}) \le \sup_{F \in \R} \mu(F,F+\varepsilon) \, ,
\end{align*}
so $\sup_{E \in \R} \mu(E,E+\varepsilon] \le \sup_{F \in \R} \mu(F,F+\varepsilon)$ and this proves equality. Similarly, on checks that $\sup_{E \in \R} \mu[E,E+\varepsilon) \le \sup_{F \in \R} \mu(F,F+\varepsilon)$, which proves (\ref{eq:sup}).

We finally prove the following. Here $\Omega = \mathcal{B}^{\mathcal{I}}$ with $\mathcal{B} \subseteq \R$ a Borel set and $\mathcal{I}$ is countable.

\begin{lem}           \label{lem:monotone}
If $\prob$ has no atoms, then any monotone $\varphi : \Omega \to \R$ which depends on finitely many $\omega_{\alpha}$ is $\mathfrak{F}_{\prob}$-measurable, where $\mathfrak{F}_{\prob}$ is the $\prob$-completion of $\mathfrak{F}$.
\end{lem}
\begin{proof}
Suppose $\varphi$ only depends on $(\omega_{\alpha})_{\alpha \in \mathcal{I}_m}$. For notational simplicity, assume $\mathcal{I}_m=\{1,\ldots,m\}$. Put $\mathcal{I}_k := \{1,\ldots,k\}$ for $1 \le k \le m$ and let $\mathfrak{F}_k$ be the $\sigma$-algebra generated by $\big(\mathop \otimes_{\alpha \in \mathcal{I}_k} \mathfrak{B}\big) \bigcup \mathcal{N}_k(\prob)$, where $\mathcal{N}_k(\prob) := \{ M \subseteq \mathcal{B}^{\mathcal{I}_k} : \prob^{\ast}(M \times \mathcal{B}^{\mathcal{I}_k^c}) = 0 \}$. Here $\prob^{\ast}$ is the outer measure defined by $\prob$ and $\mathcal{I}_k^c = \mathcal{I} \setminus \mathcal{I}_k$. Then $A \in \mathfrak{F}_k$ implies $A \times \mathcal{B}^{\mathcal{I}_k^c} \in \mathfrak{F}_{\prob}$.

Since $\varphi : \mathcal{B}^{\mathcal{I}} \to \R$ only depends on $(\omega_{\alpha})_{\alpha \in \mathcal{I}_m}$, then given $a \in \R$, $\{ \omega : \varphi(\omega) \ge a\} = A' \times \mathcal{B}^{\mathcal{I}_m^c}$ for some $A' \subseteq \mathcal{B}^{\mathcal{I}_m}$. So to show that $\varphi$ is measurable, it suffices to show that $A' \in \mathfrak{F}_m$. But if we define $\varphi_0 : \mathcal{B}^{\mathcal{I}_m} \to \R$ by $\varphi_0(\omega^m) := \varphi(\omega^m,0)$ for $\omega^m = (\omega_{\alpha})_{\alpha \in \mathcal{I}_m}$, then $\varphi_0$ is increasing and $\{ \omega^m : \varphi_0(\omega^m) \ge a \} = A'$. Thus, it suffices to show that any monotone increasing $f: \mathcal{B}^{\mathcal{I}_m} \to \R$ is $\mathfrak{F}_m$-measurable. For this, we proceed by induction, adapting an argument of Nathaniel Eldredge showing that monotone functions on $\R^m$ are Lebesgue-measurable, following \cite[Theorem 4.4]{GG}.

For $k=1$ the assertion is clear: if $f: \mathcal{B} \to \R$ is increasing and $A = \{t: f(t) \ge a\}$, then $A = \emptyset$ or $A = I \cap \mathcal{B}$ for some interval $I$. Thus, $A \in \mathfrak{B} \subset \mathfrak{F}_1$.

Now suppose $f : \mathcal{B}^{\mathcal{I}_{k+1}} \to \R$ is increasing, fix $a \in \R$ and define $g:\mathcal{B}^{\mathcal{I}_k} \to \R$ by $g(\omega^k) = \inf\{t \in \mathcal{B}: f(\omega^k,t) \ge a\}$. Then $g$ is monotone decreasing, hence $\mathfrak{F}_k$-measurable by the induction hypothesis. So by \cite[Proposition 3.3.4]{Bog07}, we have $E := \{(\omega^k,\omega_{k+1}):g(\omega^k)<\omega_{k+1}\} \in \mathfrak{F}_k \otimes \mathfrak{B}$ and $G := \{ (\omega^k,\omega_{k+1}):g(\omega^k)=\omega_{k+1} \}\in \mathfrak{F}_k \otimes \mathfrak{B}$. Moreover, for any $\omega^k \in \mathcal{B}^{\mathcal{I}_k}$ and $y \in \mathcal{B}^{\mathcal{I}^c_{k+1}}$, we have $G_{\omega^k,y} := \{ \omega_{k+1} : (\omega^k,\omega_{k+1},y) \in G \times \mathcal{B}^{\mathcal{I}_{k+1}^c} \} = \{ \omega_{k+1} : \omega_{k+1} = g(\omega^k) \} = \{g(\omega^k) \}$. We may find $F \subseteq G \times \mathcal{B}^{\mathcal{I}_{k+1}^c}$ such that $F \in \mathfrak{F}$ and $\overline{\prob}(G \times \mathcal{B}^{\mathcal{I}_{k+1}^c}) = \prob(F)$. The section $F_{\omega^k,y}$ of such $F$ is either a singleton or empty. Thus,
\[
\overline{\prob}(G \times \mathcal{B}^{\mathcal{I}_{k+1}^c}) = \prob(F) = \expect_{Y_{k+1}} \{ \mu_{\hat{\omega}_{k+1}} (F_{\omega^k,y}) \} \le \expect_{Y_{k+1}} \Big\{ \sup_{E \in \R} \mu_{\hat{\omega}_{k+1}}( E,E+\varepsilon ) \Big\}
\]
for any $\varepsilon>0$. Since $s(\prob,\varepsilon) \to 0$ as $\varepsilon \to 0$, it follows that $\overline{\prob}(G \times \mathcal{B}^{\mathcal{I}_{k+1}^c}) = 0$.

Finally, if $M=M' \times B$ with $M' \in \mathcal{N}_k(\prob)$ and $B \subseteq \mathcal{B}$, then $\prob^{\ast}(M \times \mathcal{B}^{\mathcal{I}_{k+1}^c}) \le \prob^{\ast}(M' \times \mathcal{B}^{\mathcal{I}_k^c}) = 0$, hence $\mathfrak{F}_k \otimes \mathfrak{B} \subset \mathfrak{F}_{k+1}$ and $E,G \in \mathfrak{F}_{k+1}$. But if $A = \{(\omega^k,\omega_{k+1}) :  f(\omega^k,\omega_{k+1}) \ge a\}$, then $E \subseteq A$ and $(A \setminus E) \subseteq G$. Since $E \in \mathfrak{F}_{k+1}$ and $\prob^{\ast}((A \setminus E) \times \mathcal{B}^{\mathcal{I}_{k+1}^c}) \le \prob^{\ast}(G \times \mathcal{B}^{\mathcal{I}_{k+1}^c}) = \overline{\prob}(G \times \mathcal{B}^{\mathcal{I}_{k+1}^c}) = 0$, $A \in \mathfrak{F}_{k+1}$ and the proof is complete.
\end{proof}

It is worthwile to note that the completeness of $(\Omega, \mathfrak{F}_{\prob}, \prob)$ is not only sufficient for the above argument to work, but also necessary. Indeed, there exist monotone increasing maps $f: \R^2 \to \R$ which are not Borel-measurable; see \cite{Sab0}.

\subsection*{Acknowledgements}
This work is part of my PhD thesis, supervised by Prof. Anne Boutet de Monvel and Victor Chulaevsky. I am very thankful for their encouragement and suggestions to improve this text. I also thank Prof. Konstantin Pankrashkin and Constanza Rojas-Molina for clarifying discussions on previous results, and Prof. Fr\'ed\'eric Klopp who told me about (\ref{eq:sP1}).

\providecommand{\bysame}{\leavevmode\hbox to3em{\hrulefill}\thinspace}
\providecommand{\MR}{\relax\ifhmode\unskip\space\fi MR }
\providecommand{\MRhref}[2]{%
  \href{http://www.ams.org/mathscinet-getitem?mr=#1}{#2}
}
\providecommand{\href}[2]{#2}

 
\end{document}